\documentclass[journal]{IEEEtran}

\usepackage{latexsym}
\usepackage{amssymb}
\usepackage{bm}
\usepackage[dvips]{graphics}
\usepackage{graphicx}
\usepackage{psfrag}
\usepackage{amsfonts,amsmath,amssymb,hyperref}
\usepackage{algorithm}
\usepackage{algorithmic}
\usepackage{dsfont,color,subfigure,setspace,psfrag,cite}

\usepackage{xspace}
\usepackage{bbm}

\newcommand{\Ac}{\mathcal{A}}

\newcommand{\Ec}{\mathcal{E}}

\newcommand{\Sc}{\mathcal{S}}

\newcommand{\Uc}{\mathcal{U}}

\newcommand{\Xc}{\mathcal{X}}
\newcommand{\Yc}{\mathcal{Y}}

\newcommand{\uh}{{\hat{u}}}

\def\a{\alpha}

\def\g{\gamma}

\def\l{\lambda}

\DeclareMathOperator\E{E}
\let\P\relax
\DeclareMathOperator\P{P}

\newcommand\ie{i.e.,\xspace}
\def\textiid{i.i.d.\@\xspace}
\newcommand\iid{\ifmmode\text{ i.i.d. } \else \textiid \fi}

\newcommand{\ind}{\mathbbmss{1}}

\newtheorem{remark}{Remark}

\newtheorem{theorem}{Theorem}
\newtheorem{lemma}{Lemma}

\begin{document}

\title{Outage Analysis of  Uplink  Two-tier Networks}

\author{Zolfa Zeinalpour-Yazdi~\IEEEmembership{}
        and~Shirin Jalali~\IEEEmembership{}
      \thanks{Z. Zeinalpour-Yazdi is with the Department of Electrical and computer Engineering, Yazd University, Yazd, Iran (e-mail: zeinalpour@yazd.ac.ir),}
      \thanks{S. Jalali is with the Department of   Electrical  Engineering, Princeton university, NJ 08540 (e-mail: sjalali@princeton.edu)}}

 \maketitle

\newcommand{\p}{\mathds{P}}
\newcommand{\Lc}{\mathcal{L}}
\newcommand{\mb}{\mathbf{m}}
\newcommand{\bb}{\mathbf{b}}
\newcommand{\Xb}{\mathbf{X}}
\newcommand{\Yb}{\mathbf{Y}}
\newcommand{\Ub}{\mathbf{U}}
\newcommand{\La}{\Lambda}
\newcommand{\su}{\underline{s}}
\newcommand{\xu}{\underline{x}}
\newcommand{\yu}{\underline{y}}
\newcommand{\Xu}{\underline{X}}
\newcommand{\Yu}{\underline{Y}}
\newcommand{\Uu}{\underline{U}}
\newcommand{\ex}{{\rm e}}
\newcommand{\SIR}{{\rm SIR}}

\begin{abstract}
Employing multi-tier  networks is among   the most promising approaches to address the rapid growth of the data demand in cellular networks. In this paper, we study a two-tier uplink cellular network consisting of femtocells and a macrocell. Femto base stations, and  femto and macro users are assumed to be spatially deployed based on independent  Poisson point processes. We  consider an open access assignment policy, where each macro user based on the ratio between  its distances from its nearest femto access point (FAP) and  from the macro base  station (MBS) is  assigned to either of them. By tuning the threshold, this policy  allows controlling the coverage areas of FAPs.  For a fixed threshold, femtocells coverage areas  depend on their distances from the MBS; Those closest to the fringes will have the largest coverage areas. Under this  open-access policy, ignoring the additive noise, we  derive analytical upper and lower bounds on the outage probabilities  of  femto users and macro users that are subject to fading and path loss. We also study the effect of the distance from the MBS on the outage probability experienced by the users of a femtocell.  In all cases, our simulation results comply with our analytical bounds.
\end{abstract}

\section{Introduction}
Wireless cellular networks, originally designed for voice communications, are nowadays commonly used for surfing the Internet or communicating image, audio or video files. This massive unpredicted overhead load  has urged communication engineers to develop new approaches to design and employment of  cellular communication systems. One of such relatively new techniques, which has been proved to be successful,  is employing multi-tier networks. For instance, in the case of  two-tier networks, the existing cellular  network is overlaid by {\em femtocells}, which are employed by users in an ad-hoc manner at their homes or offices.

Analytical performance evaluation  of  cellular networks has  always been a complicated task.  Modeling various aspects of cellular networks, such as the  physical channel itself,   has been a cornerstone of such analysis and therefore    the subject of  extensive  research for many years.   Modeling the users' and cells' locations is another  aspect of a cellular network that also plays a major role in analytical evaluations. Traditionally, the idealized  grid model has been  employed  to model the locations of the cells and their coverage areas. This model, although  simple to describe, is intractable for most analytical evaluations and is also arguably  not very accurate. This is especially true  in heterogenous networks with ad hoc employment of  small cells.  More comprehensive and recent models for spatial distributions of the cells and users are  models based on stochastic geometric tools such as Poisson point process  (PPP). Such models are  advantageous  from  two main  perspectives: first, they provide a more realistic model of cellular networks compared to the traditional grid-based  models, second, they make  the analysis  more tractable.

In this paper, we analyze the outage performance of an uplink two-tier network with a MBS located at the center of  a circle representing  its coverage area; macro users (MU), femto users (FU) and FAPs are assumed to be spatially distributed within the circle  randomly and independently according to PPPs with different densities. We  consider the \emph{open access policy} studied in \cite{CheungQ:12,JoS:12} for downlink communication. This model covers closed access policy as a special case and  allows optimizing the coverage areas of the FAPs when  the system parameters   vary.  Using this model, we derive tight upper and lower bounds on the outage probabilities of users covered by the MBS and also FUs and MUs that are covered by FAPs. To achieve this goal we first derive upper and lower bounds on the Laplace transform of the number of MUs serviced by the base station. We also derive the   Laplace transform  of the number of MUs covered by a FAP located at a specific distance from the MBS. Employing the Laplace transforms of the number of users in each  group, plus some geometric analysis, we bound our desired outage probabilities.     Our simulation results confirm  and comply with our bounds.

\subsection{Related work}

While employing PPP as a stochastic model for users or access points distributions was originally proposed in 1997 in  \cite{BaccelliZ:97,BaccelliK:97,Brown:00},  it was not until recently that this model was used for analyzing the performance of wireless cellular networks. (Refer to \cite{HaenggiA:09} for a review of this model.)
The stochastic geometric-based models such as PPP was  employed by Baccelli et al.~in \cite{BaccelliM:09} to analyze large mobile ad hoc network (MANET) and by Andrews et al.~in \cite{AndrewsB:09} to study  the downlink performance of cellular networks. Later, this model was  used for analyzing the downlink performance of multi-tier networks  \cite{DhillonG:12, Mukherjee:12,WangQ:12}. However, as  mentioned in \cite{AndrewsC:12},  similar analysis for uplink communication has been missing until very recently.

Multi-tier networks have been studied from different perspectives  such as power control \cite{ChandrasekharA:09-power,JoM:09}, spectrum allocation \cite{ChandrasekharA:09-tcom,ChaZ:13}, and  exploiting cognitive radio techniques \cite{XiangZ:10,ShiH:10}. These are just few examples of some related work and by no means are meant to be a comprehensive review of the literature. (See \cite{ChandrasekharA:08, ClaussenH:08, AndrewsC:12} and the references therein for a relatively  comprehensive   review of the literature.)

 Analytic study of the outage performance of a single-tier network with nodes  distributed according to a PPP  is done in \cite{NovlanD:13}. Uplink performance  of  two-tier networks has been studied  in the literature under different models and approximations.  While most of the work on this topic has been on traditional grid model, recently there has been several results on analyzing uplink performance of two-tier networks under PPP model for users and access points. Chandrasekhar et al.~ study  outage probabilities of femto and MUs  distributed according to PPPs in a reference macrocell in  \cite{ChandrasekharA:09}. The authors consider a  CDMA-based model under closed access and approximate the outage probability. Xia et al. in \cite{XiaC:10} compare closed access versus open access policy in an uplink communication. In their analysis, they consider a reference macrocell with the base station located at the center, and a single FAP located at  a specific distance  from the base station. The MUs are assumed to be distributed independently at random. They suggest that while  for orthogonal multiple access schemes such as TDMA or OFDMA the choice of open versus closed depends on the users density, in non-orthogonal schemes such as CDMA open access is strictly better than closed access.

The uplink performance of macrocells overlaid with femtocells is also studied  in \cite{ChakchoukH:12}. There, while the authors consider PPP spatial distribution for MUs, FUs and FAPs, the users  assignment policy is closed access, and by assuming a TDMA scheme they limit the number of active users in each femtocell per  time slot to one. In \cite{YuM:12}, the authors study the distribution of the  signal to interference plus
noise ratio (SINR)  in both uplink and downlink, when  time division duplex (TDD) is employed. In their setup, the users of each tier are distributed according to a PPP and each user connects to the closest base station.

In an independent  work, which the authors became aware of right before submitting this paper, Bao et al.~analyze the interference and outage performance of  a two-tier uplink network under closed access policy \cite{BaoL:13}.  The authors of \cite{BaoL:13} also study the open access policy in  a subsequent paper \cite{BaoL:13-acm}. While the ultimate goals in \cite{BaoL:13-acm} and this paper are  the same, there are some major differences betweens the two.  First, unlike this paper,  in \cite{BaoL:13-acm}, each femtocell is assumed to have a fixed coverage area, and a MU is handed off to the FAP if it falls within that fixed coverage area. Here, we consider a different open access policy, where each MU decides to connect to either its closest FAP or the MBS, based on its distances from them. This policy leads to FAPs having different coverage areas, depending on their distances to the MBS. This assignment policy introduces new geometrical aspects to our outage analysis. Second, unlike \cite{BaoL:13-acm}, we derive closed-form expressions for our  upper and lowers bounds on the outage probabilities of MUs and FUs. For a  MU serviced by the MBS, we study and bound its outage performance as a function of its distance of the FAP from the MBS.
 Finally, here we consider multi-carrier frequency hopping modulation, which provides a decentralized alternative to OFDM.  In \cite{BaoL:13-acm}, the authors consider a single shared  channel for all users. Finally, one of the reviewers pointed us to the work of ElSawy et al.~\cite{ElSawyH:14}, which has appeared on Arxiv after  our initial submission. In \cite{ElSawyH:14}, the authors study the uplink performance of a multi-tier network under a different access policy where each user connects to its closest access point (femto or macro).

\subsection{Notation}
Calligraphic letters such as $\Xc$ and $\Yc$ represent  sets. The size of set $\Xc$ is denoted by $|\Xc|$.
Given sample space $\Omega$ and  event $\Ec\subseteq \Omega$, $\ind_{\Ec}$ is an indicator random variable that is one when  event $\Ec$ happens.  For $1\leq i\leq j \leq n$, $x_i^j \triangleq (x_i,x_{i+1},\ldots,x_j)$. Also, for simplicity  $x^i=x_1^i$. Uppercase letter characters such as $X$ and $Y$ are  used for  matrices and random variables.

\subsection{Paper organization}
The organization of the paper is as follows. Section \ref{sec:model} reviews the network model studied in this paper from various perspectives: modulation technique, spatial distributions of  users, channel model and access policy. In Section \ref{sec:dist}, we study the users density distributions. The results of this section is used extensively in the following sections in analyzing the performance of the system.   In Section \ref{sec:outage}, which contains the main results of the paper, we analyze the femto and MUs outage probabilities. Section \ref{sec:numerical} presents the simulation results and compares them with our analytical bounds. Finally, Section \ref{sec:conclusion} concludes the paper.

\section{System model}\label{sec:model}

\subsection{MCFH technique}

Orthogonal frequency devision multiplexing (OFDM) is  a widely popular  multiple access method in wireless networks, and has received a lot of attention in recent years. In an OFDM-based multiple access  system, the  carrier frequencies are assigned  by the central node. (Three different methods for assigning frequencies are described in \cite{Lawrey99}.) However, this centralized frequency assignment  is not quite desirable  for  emerging decentralized wireless cellular  networks such as femtocells, where, due to practical challenges', it is preferred   to minimize the coordination between  the central  and the femto base stations.

Multicarrier frequency-hopping  (MCFH) modulation introduced in \cite{LanceK:97}, and later analyzed by various researchers  \cite{EbrahimiN:04, YazdiN:04,TaghaviN:03,NikjahB:08}, provides  a decentralized alternative to OFDM modulation. In  MCFH, similar to OFDM,  all  sub-carriers are orthogonal to each other. However, unlike OFDM,  in a multi-user setup, the carriers are not assigned to the users by a central node, and the users  are allowed to randomly and independently  select their carriers.
In addition to being decentralized, another advantage of  MCFH to OFDM, as  will become clear throughout the paper, is that it makes the model more amenable to direct analysis. The results of such analysis will provide insight on how to select the systems' parameters in an OFDM-based system as well.
In this paper, we assume that all users adopt MCFH modulation. While MCFH is clearly  different from  OFDM,   most of our results will continue to hold for OFDM-based systems with some mild adjustments.

In MCFH, the available bandwidth is divided into $n_s$ non-overlapping adjacent subbands. Each subband respectively is divided into  $n_h$  equispaced frequencies. Hence, overall,  there will be  $n_sn_h$ available \emph{subchannels}. (It is usually said that   the system's processing gain ($G$) is equal to $n_sn_h$.) At each time, each user  uniformly at random  selects one of the $n_h$ carriers  in each subband.  Fig.~\ref{fig:mcfh} shows the carrier selections  in a  simple MCFH system with $n_s=3$ and $n_h=4$ and two users.  As shown in the figure, since unlike OFDM,  users select their carriers independently with no coordination, it is possible that two users send data over the same frequency  simultaneously.

\begin{figure}
\begin{center}
\includegraphics[width=6.5cm]{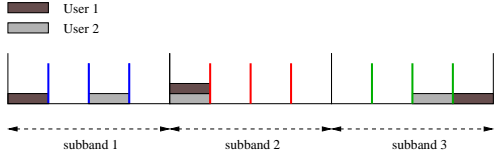}\caption{Depiction of MCFH frequency assignments.}\label{fig:mcfh}\end{center}
\end{figure}

\subsection{Spatial distribution}

For spatial distribution of  MUs, FUs and FAPs, we follow the model introduced in \cite{CheungQ:12}. We consider a MBS $b_m$ located at the center of a circle of radius $R$ denoted by $\Sc_m$.    FAPs $\Ac_f$ are distributed according to a PPP with density $\l_f$. Therefore, the number of FAPs ($|\Ac_f|$) is distributed as ${\rm Poiss}(\bar{n}_{\rm fap})$, where $\bar{n}_{\rm fap}\triangleq\pi R^2\l_f$. Conditioned on $|\Ac_f|=m$, the locations of the $m$ FAPs are uniformly distributed over $\Sc_m$.  Independently, MUs $\Uc_m$ are distributed based on a PPP with density $\mu_m$. Note that  ``MUs'' are users that are not inside a home, office, etc.~that is equipped with a FAP. However, a MU might be served by a FAP based on its distance from the MBS and the locations of surrounding FAPs.  Finally, FUs of FAP $a_f\in{\cal A}_f$ are  distributed according to a PPP with density $\mu_f$   restricted to a ring of internal radius $r_f$ and width $\Delta$ centered at $a_f$.  For FAP $a_f\in\Ac_f$, let $\Uc_f(a_f)$  and $\Uc_m(a_f)$ denote the set of  FUs and MUs serviced by the FAP $a_f$, respectively. Clearly, $\cup_{a_f\in\Ac_f}\Uc_m(a_f)\subseteq \Uc_m.$

Various studies indicate that open access policies have superior performance both from the perspective of the FUs (in uplink) and MUs (in downlink). Therefore, in this paper we focus on a two-tier network with open access policy. The specific access policy that we consider is described in Section \ref{section:access}.

\begin{remark}
In our analysis we consider a single MBS located at the center of a circle of radius $R$. In reality of course there are more MBSs. The placement of the macrocells can be modeled either as a deterministic process or random based on an independent PPP with density $\lambda_m$. In both cases, it is reasonable to assume that each MU connects to its closest MBS, and hence divide the plane based on the Voronoi partition determined by the locations of  MBSs.  Assuming that the MBSs employ one of the known frequency reuse methods, are hence orthogonalize the users of neighboring macrocells, then, without loss of generality, in the analysis one can  focus on the  case where there is only one MBS. In a random setting, where MSBs are employed according to a PPP of density $\lambda_m$, it is proved in [35] that the expected number of FAPs in a ``typical'' macrocell becomes equal to $\lambda_f/\lambda_m$. Based on this result, choosing $\lambda_m=1/(\pi R^2)$, the expected number of FAPs in a ``typical'' macrocell is equal to $\lambda_f/\lambda_m=\pi R^2 \lambda_f$, which is consistent by our model in this paper. By controlling radius $R$, we can  study the effect MBSs' density $\lambda_m$ on the performance.
\end{remark}

\subsection{Channel Model}
We consider  both small scale fading and path loss. Let $h^i_{u,a_f}$ and $h^i_{u,b_m}$ denote the fading coefficients corresponding to the channel in subband $i\in [1:n_s]$ from user $u$ to  FAP $a_f$ and to MBS $b_m$, respectively.  We consider a slow-fading channel model, and assume that  the fading coefficients remain constant during the whole coding block. Furthermore, we assume that the coefficients corresponding to different subbands and also different channels are all independent.  The channel coefficients  are assumed to have a  Rayleigh distribution. That is, the  power attenuation coefficient  $|h_{u,a}^i|^2$, where  $a\in\{a_f,b_m\}$, is exponentially distributed as $\P\left(|h_{u,a}^i|^2 >x\right)= {\rm e}^{- x /{\sigma}^2}$,
for $x\geq 0$.  The path loss affecting the  signal  transmitted by user $u$ to base station $a$, $a\in\{a_f,b_m\}$, is modeled as $\texttt{PL}_{u,a}=L_0 d_{u,a}^{\alpha},$
where $L_0$ is path loss at unit distance, and $\alpha>2$ denotes the attenuation factor \cite{JoS:12}.

\vspace{-1em}
\subsection{Access policy}\label{section:access}

Consider macro user $u\in\Uc_m$. For user $u$ and FAP or MBS $a$, let $d(u,a)$ denote their Euclidian distance. Further, let  $d_{u}^{(f)}$ denote the distance between user $u$ and its nearest  FAP, i.e.,  $d_{u}^{(f)}\triangleq \min\{d(u,a_f): a_f\in\Ac_f\}$. As mentioned earlier, we focus on open access policy, where MUs can also be serviced  by FAPs. We consider the following open access  policy, which was considered in  \cite{CheungQ:12}. Let  $\kappa<1$ be a parameter of the system. Then, according to this assignment policy,
\begin{enumerate}
\item if $d_{u_m}^{(f)}<\kappa d(u_m,b_m)$, then MU $u_m$ is assigned to its closest FAP,
\item if $d_{u_m}^{(f)}\geq \kappa d(u_m,b_m)$, then MU $u_m$ is assigned to the MBS $b_m$.
\end{enumerate}
Letting $\kappa=0$, requires all MUs  to be serviced by the base station, which is   equivalent to having a  closed access assignment policy.  $\kappa$ controls the coverage areas of FAPs and  increasing it enlarges the coverage areas.

As defined earlier, $\Uc_m(a_f)\subset\Uc_m$  denotes the set of MUs that are serviced by FAP $a_f\in\Ac_f$. Let $\Uc_m^{(-f)}$ denote all  MUs that are not  serviced by FAPs, \ie $\Uc_m^{(-f)}=\Uc_m\backslash (\cup_{a_f\in\Ac_f}\Uc_m(a_f))$.

\section{Users density distribution}\label{sec:dist}

Before analyzing the signal to interference ratios (SIR)  experienced by different users in different groups, in this section we study the distributions of $N_f^{a_f}\triangleq |\Uc_f(a_f)|$, $N_m^{a_f}\triangleq |\Uc_m(a_f)|$ and $N_m^{b_m}\triangleq |\Uc_m^{(-f)}|$. By our assumption, the FUs are distributed in a ring of width $\Delta$ and internal radius $r_f$. Hence, $ N_f^{a_f}\sim {\rm Poiss}(\bar{n}_{\rm fu})$, where $\bar{n}_{\rm fu}\triangleq \pi((r_f+\Delta)^2-r_f^2)\mu_f$ denotes the expected number of FUs in each FAP, and  its Laplace transform $\Phi_{N_f^{a_f}}(s)$ is equal to \vspace{-1em}
\begin{align}
\Phi_{N_f^{a_f}}(s)\triangleq \E[\ex^{-s N_f^{a_f}}]=\ex^{\bar{n}_{\rm fu}(\ex^{-s}-1)}.\label{eq:Phi-Uf-af}
\end{align}
Also, let $\bar{n}_{\rm mu}\triangleq \pi R^2 \mu_m$ denote the expected number of all  MUs.
\begin{figure}
\begin{center}
\psfrag{bm}{ $b_m$}
\psfrag{af}{ $a_f$}
\psfrag{d}{ $r$}
\psfrag{r}{ $\hspace{-0.14cm}r_c$}
\includegraphics[width=6cm]{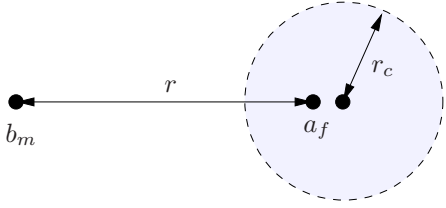}\caption{MUs served by FAP $a_f$ located at distance $r$ from MBS $b_m$.}\vspace{-2.5em}\label{fig:MU-covered-FAP}\end{center}
\end{figure}

Consider  MBS $b_m$ and FAP $a_f$ shown in Fig.~\ref{fig:MU-covered-FAP} that are located at distance $r$ from each other. MU $u_m$ is served by FAP $a_f$ instead of $b_m$, if $d(u_m,a_f)<\kappa d(u_m,b_m)$, where $\kappa<1$. Translating this condition into cartesian coordinate dimensions with the origin located at $b_m$ and the $x$-axis along the line connection $b_m$ to $a_f$, we obtain $(x-r)^2+y^2\leq \kappa^2 (x^2+y^2)$,
or
$(1-\kappa^2)x^2+(1-\kappa^2)y^2 -2rx+r^2\leq 0.$ In other words,
\[
(x-{r\over 1-\kappa^2})^2+y^2\leq r^2\Big({1\over (1-\kappa^2)^2}-{1\over 1-\kappa^2}\Big),
\]
which is equivalent to a circle of radius $r_c={\kappa r\over 1-\kappa^2}$
centered at $(r/(1-\kappa^2),0)$. Therefore,  the coverage area of each FAP depends on its distance from the base station. As the distance increases, the coverage area, and  the expected number of covered MUs  increase as well. In summary, given FAP $a_f$ located at distance $d_f=d(b_m,a_f)$, the coverage area of $a_f$ for MUs, \ie the area in which MUs are serviced by FAP $a_f$ is a circle of radius $\sqrt{\gamma} d_f$, where \vspace{-.5em}
\begin{align}
\g\triangleq {\kappa^2\over (1-\kappa^2)^2}.\label{eq:gamma-def}
\end{align}
Therefore, $N_m^{a_f}\sim {\rm Poiss}( \pi \g d_f^2\mu_m)$, where $\bar{n}_{\rm mu}^f\triangleq\pi\gamma d_f^2 \mu_m$, and \vspace{-.5em}
\begin{align}
\Phi_{N_m^{a_f}}(s|d_f)& \triangleq  \E[\ex^{-s N_m^{a_f}}|d(a_f,d_f)=d_f]=\ex^{\bar{n}_{\rm mu}^f(\ex^{-s}-1)}.\label{eq:Phi-Um-af}
\end{align}
In some cases, in our analysis we are interested in the distribution of $N_m^{a_f}-1$, conditioned on  $N_m^{a_f}\geq 1$. For that we define,\vspace{-.5em}
\begin{align}
\Phi^+_{N_m^{a_f}}(s|d_f)&\triangleq \E\big[\ex^{-s (N_m^{a_f}-1)}\big|d(a_f,d_f)
=d_f,N_m^{a_f}\geq1\big]\nonumber\\
&=({\ex^{\bar{n}_{\rm mu}^f(\ex^{-s}-1)}-\ex^{-\bar{n}_{\rm mu}^f}\over 1-\ex^{-\bar{n}_{\rm mu}^f}})\ex^s.\label{eq:Phi-Um-af-plus}
\end{align}

Finally, we study $N_m^{b_m}$. Conditioned on $\Ac_f$ (locations of FAPs),  $N_m^{b_m}$ is a Poisson random variable of mean $\mu_m S_{-f}$, where $S_{-f}$ denotes the area exclusively  covered only by $b_m$ and not FAPs. Therefore, $\Phi_{N_m^{b_m}}(s)=\E[\ex^{\mu_mS_{-f}(\ex^{-s}-1)}].$
 \begin{theorem}\label{thm:1}
For $s\geq 0$, the Laplace transform  of $N_m^{b_m}$, $ \Phi_{N_m^{b_m}}(s)$, satisfies \vspace{-.5em}
\[
\Phi_{N_m^{b_m}}(s) \geq \ex^{ \bar{n}_{\rm mu}(\ex^{-s}-1)}\vspace{-.5em}
\]
and $\Phi_{N_m^{b_m}}(s) \leq \ex^{\bar{n}_{\rm mu}(\ex^{-s}-1) + \bar{n}_{\rm fap} (\tau(s)-1)} +\ex^{\g^{-1}-\bar{n}_{\rm fap} +\g^{-1}\log (\g \bar{n}_{\rm fap})},$  where
\begin{align}
\tau(s)\triangleq {\ex^{(1-\ex^{-s})\g \bar{n}_{\rm mu}}-1 \over (1-\ex^{-s})\g \bar{n}_{\rm mu}},\vspace{-.5em}\label{eq:tau-s}
\end{align}
and $\g$ is defined in \eqref{eq:gamma-def}.
\end{theorem}
\begin{proof} Since $N_m^{b_m}\leq  |\Uc_m|$ always holds,  for $s\geq 0$, $\ex^{-sN_m^{b_m}}\geq \ex^{-s|\Uc_m|}$. Therefore, $ \Phi_{N_m^{b_m}}(s)\geq \E[ \ex^{-s|\Uc_m|}]=\ex^{ \bar{n}_{\rm mu}(\ex^{-s}-1)}$. If the  coverage areas of the FAPs do not overlap, then $S_{-f}=\pi(R^2-\g\sum_{a_f\in\Ac_f}d^2(a_f,b_m))$. In general, the regions might overlap, and therefore $S_{-f}\geq  \pi(R^2-\g\sum_{a_f\in\Ac_f}d^2(a_f,b_m))$. This lower bound  clearly is a function of the locations of the FAPs, and can be negative. Let $\Ec$ denote the event that $|\Ac_f| \leq \gamma^{-1}$. However, if $\Ec$ holds, then $\g\sum_{a_f\in\Ac_f}d^2(a_f,b_m) \leq \g R^2 |\Ac_f| \leq R^2,$ and $ \pi(R^2-\g\sum_{a_f\in\Ac_f}d^2(a_f,b_m))\geq 0$. We employ this observation to derive an upper bound on $\Phi_{N_m^{b_m}}(s)$. By the law of total expectation
\begin{align}
\Phi_{N_{m}^{b_m}}(\hspace{-.1em}s\hspace{-.1em}) \hspace{-.4em} &\; = \;  \hspace{-.4em}\E[\ex^{\mu_mS_{-\hspace{-.1em}f}(\ex^{-s}\hspace{-.3em}-1\hspace{-.1em})}|\Ec]\hspace{-.2em}\P\hspace{-.1em}(\hspace{-.1em}\Ec\hspace{-.1em})
\hspace{-.3em}
+\hspace{-.2em}\E[\ex^{\mu_mS_{-\hspace{-.1em}f}(\ex^{-s}\hspace{-.2em}-1)}|\Ec^c]\hspace{-.2em}\P\hspace{-.1em}(\hspace{-.1em}\Ec^c\hspace{-.1em})\nonumber\\
&\; \leq  \;  \E[\ex^{\mu_mS_{-f}(\ex^{-s}-1)}|\Ec]\P(\Ec)+\P(\Ec^c).\label{eq:step-11}
\end{align}
On the other hand,
\begin{align}
\E&[\ex^{\mu_mS_{-f}(\ex^{-s\hspace{-.1em}}-1)}|\Ec] \leq \E[\ex^{\mu_m\pi(R^2-\g\sum\limits_{a_f\in\Ac_f}d^2(a_f,b_m)) (\ex^{-s}\hspace{-.1em}-1)}|\Ec] \nonumber\\
&\;= \;{\ex^{\bar{n}_{\rm mu}(\ex^{-s}-1)} \over P(\Ec)} \sum_{n=0}^{\lfloor \g^{-1} \rfloor} \E[\ex^{-\mu_m\pi \g (\ex^{-s}-1)\sum\limits_{i=1}^nd_i^2 }] \P(|\Ac_f|=n) \nonumber\\
&\; \overset{{\mathrm{(a)}}}{=} \;{\ex^{\bar{n}_{\rm mu}(\ex^{-s}-1)} \over P(\Ec)} \sum_{n=0}^{\lfloor \g^{-1} \rfloor} \tau^n(s) \P(|\Ac_f|=n) \nonumber\\
&\;\leq  \;{\ex^{\bar{n}_{\rm mu}(\ex^{-s}-1)} \over P(\Ec)} \sum_{n=0}^{\infty} \tau^n(s) \P(|\Ac_f|=n) \nonumber\\
&\;={\ex^{\bar{n}_{\rm mu}(\ex^{-s}-1)} \over P(\Ec)} \ex^{\bar{n}_{\rm fap} (\tau(s)-1)},\label{eq:step-12}
\end{align}
where $(a)$ holds because
\begin{align*}
\E[\ex^{(1-\ex^{-s})\mu_m \pi\g d^2({a}_f,b_m)}]&=\int_{0}^R \ex^{(1-\ex^{-s})\mu_m \pi\g r^2}{2r\over R^2}dr\nonumber\\
& = {\ex^{(1-\ex^{-s})\g \bar{n}_{\rm mu}}-1 \over (1-\ex^{-s})\g \bar{n}_{\rm mu}}=\tau(s).
\end{align*}
By the Chernoff bound, for $x>0$, $\P(\Ec^c)=\P(|\Ac_f| > \g^{-1})  \leq  {\E[ \ex^{x|\Ac_f|}]\over \ex^{x/\g}}= {\ex^{\bar{n}_{\rm fap} (\ex^x-1)}\over \ex^{x/\g}}.$
Optimizing the bound by  choosing $x$  as the solution of   $\bar{n}_{\rm fap} \ex^x-\g^{-1}=0$, we obtain
\begin{align}
\P(\Ec^c)   &\;\leq  \; \ex^{\g^{-1}-\bar{n}_{\rm fap} + \g^{-1}\log (\g \bar{n}_{\rm fap})}.\label{eq:step-13}
\end{align}
Combining \eqref{eq:step-11}, \eqref{eq:step-12} and \eqref{eq:step-13} yields the desired result.
\end{proof}


In our outage analysis presented in the proceeding sections, in some cases we study  the case where we know that there exists one FAP $a_f$ at distance $d_f$ from $b_m$. In those cases, it will be useful to define the conditional Laplace transform of $N_m^{b_m}$ as $\Phi_{N_m^{b_m}}(s|d_f)$.

\begin{theorem}\label{thm:2}
For $s\geq 0$ and $d_f\in(0,R)$, we have
$ \Phi_{N_m^{b_m}}(s|d_f) \geq \ex^{ \bar{n}_{\rm mu}(\ex^{-s}-1)},$
and
\begin{align*}
\Phi_{N_m^{b_m}}(s|d_f) \; \leq &\;   \;{\ex^{(\bar{n}_{\rm mu}-\bar{n}_{\rm mu}^f)(\ex^{-s}-1)} \over (1-\ex^{-\bar{n}_{\rm fap}  })\tau(s)} \ex^{\bar{n}_{\rm fap} (\tau(s)-1)}\nonumber\\
&+{\ex^{\g^{-1}-\bar{n}_{\rm fap} + \g^{-1}\log (\g \bar{n}_{\rm fap})} \over 1-\ex^{-\bar{n}_{\rm fap}}},\vspace{-.5em}
\end{align*}
where $\tau(s)$ and $\g$ are defined in \eqref{eq:tau-s}, and \eqref{eq:gamma-def}, respectively,  and as before $\bar{n}_{\rm mu}^f=\pi \g d_f^2 \mu_m$, $\bar{n}_{\rm mu}=\pi R^2 \mu_m$, and $\bar{n}_{\rm fap}=\pi R^2\l_f$.
\end{theorem}
\vspace{0.5em}
\begin{proof}
The proof closely follows the steps of the proof of Theorem \ref{thm:1}. The only  difference is that here we know that $|\Ac_f|\geq 1$ and that one FAP is located at distance $d_f$ from $b_m$. Therefore, $S_{-f}=\pi(R^2-\g d_f^2-\g\sum_{a'_f\in\Ac_f\backslash a_f}d^2(a'_f,b_m))$.  Define event $\Ec$ as before. Then, again by the law of total expectation, \vspace{-.5em}
\begin{align*}
\Phi_{N_m^{b_m}}(s|d_f)
 & \leq  \E[\ex^{\mu_mS_{-f}(\ex^{-s}-1)}|\Ec, |\Ac_f|>0]\P(\Ec||\Ac_f|>0)\nonumber\\
 &\;\;+\P(\Ec^c||\Ac_f|>0).
\end{align*}
Note that
\begin{align}
\E&[\ex^{\mu_mS_{-f}(\ex^{-s}-1)}|\Ec,d_f]\nonumber\\
 &\;\leq\;  \hspace{-.2em}\E[\ex^{\mu_m\pi(R^2-\g d^2_f- \sum_{a'_f\in\Ac_f\backslash a_f}d^2(a'_f,b_m)) (\ex^{-s}-1)}|\Ec,d_f] \nonumber\\
&\;= \;{\ex^{(\bar{n}_{\rm mu}-\bar{n}_{\rm mu}^f)(\ex^{-s}-1)} \over P(\Ec|  |\Ac_f|>0)} \sum_{n=1}^{\lfloor \g^{-1} \rfloor} \E[\ex^{-\mu_m\pi \g (\ex^{-s}-1)\sum\limits_{i=1}^{n-1}d_i^2 }]\nonumber\\
&\hspace{4.5cm}. \P(|\Ac_f|=n| |\Ac_f|>0) \nonumber\\
&\; =\hspace{-.3em} \;{\ex^{(\bar{n}_{\rm mu}-\bar{n}_{\rm mu}^f)(\ex^{-s}-1)} \over P(\Ec| |\Ac_f|>0)(1-\ex^{-\bar{n}_{\rm fap}  })}\hspace{-.4em} \sum_{n=1}^{\lfloor \g^{-1} \rfloor} \tau^{n-1}(s) \P(|\Ac_f|=n) \nonumber\\
&\;\leq  \;{\ex^{(\bar{n}_{\rm mu}-\bar{n}_{\rm mu}^f)(\ex^{-s}-1)} \over P(\Ec| |\Ac_f|>0)(1-\ex^{-\bar{n}_{\rm fap} })} \sum_{n=0}^{\infty} \tau^{n-1}(s) \P(|\Ac_f|=n) \nonumber\\
&\;\leq  \;{\ex^{(\bar{n}_{\rm mu}-\bar{n}_{\rm mu}^f)(\ex^{-s}-1)} \over P(\Ec| |\Ac_f|>0)(1-\ex^{-\bar{n}_{\rm fap} })\tau(s)} \ex^{\bar{n}_{\rm fap} (\tau(s)-1)}.
\end{align}
Furthermore, $\P(\Ec^c||\Ac_f|\geq 1) \;  = \;{\P(\Ec^c)\over \P(|\Ac_f|\geq 1)} \leq {\ex^{\g^{-1}-\bar{n}_{\rm fap} + \g^{-1}\log (\g \bar{n}_{\rm fap})} \over 1-\ex^{-\bar{n}_{\rm fap}}}$,
where the last step follows from \eqref{eq:step-13}.
\end{proof}

\section{Outage analysis} \label{sec:outage}
In this section we analyze the outage performance of MUs and FUs in the uplink network  described in Section \ref{sec:model}.  We assume that every user equipment employs power control to compensate for the effect of path loss. By power control, MUs serviced by the MBS intend to achieve received power levels of $P_m$. Similarly, FUs and MUs serviced by  FAPs adjust their transmitted power to achieve received power of  $P_f$.   We further assume that  the performance of   the users is primarily limited by the interference caused by other users of both tiers. Therefore we ignore the effect of additive Gaussian noise in our analysis.

To bound the outage probability, in each case, we first compute the signal to interference ratio (SIR) experienced by user equipments. The derived SIRs are probabilistic and depend on channel coefficients, and users locations. Then, we bound  the outage probabilities by employing results we proved in the previous section.
\subsection{MU served by a FAP}
Consider FAPs $a_f$ and   $\hat{a}_f\in\Ac_f\backslash a_f $. The distance between user $u$ that is covered by $\hat{a}_f$ is usually much smaller than the distance between $u$ and $a_f$. That is, $d(u,\hat{a}_f)\ll d(u,a_f)$, or ${d(u,\hat{a}_f)\over d(u,a_f)}\ll 1$.  Therefore, in evaluating the performance of users covered by $a_f$,  unless the density of FAPs ($\lambda_f$) is very large, the term  corresponding to the interference caused by users (macro or femto) covered by other FAPs is negligible compared to the other terms.  Making this approximation, the upload SIR experienced  by user $u_m\in\Uc_m(a_f)$ in subband $i\in \{1,2,\ldots, n_s\} $ is equal to
\begin{align}
\SIR_{m,f}&= {\frac{P_f|h^i_{u_m,a_f}|^2}{n_s} \over I_{m,f}},\label{eq:SIR-m-fap}
\end{align}
where
\begin{align}
I_{m,f}=&  \sum\limits_{u_f\in \Uc_f(a_f)} \hspace{-.5em}\frac{P_f|h^i_{u_f,a_f}|^2}{G}+\hspace{-1em}\sum\limits_{\uh_m\in \Uc_m(a_f)\backslash u_m}\hspace{-.5em}\frac{P_f|h^i_{\uh_m,a_f}|^2}{G}\nonumber\\
&+\hspace{-1em}\sum\limits_{\uh_m\in \Uc_m^{(-f)}}\hspace{-.3em}\Big({d(\uh_m,b_m)\over d(\uh_m,a_f)}\Big)^{\a}\frac{P_m|h^i_{\uh_m,a_f}|^2}{G}. \label{eq:I_macro}
\end{align}
In \eqref{eq:I_macro}, the interference terms are caused by the FUs of FAP $a_f$, the other  MUs of FAP $a_f$, and the MUs serviced by the MBS, respectively.   In our model, from the perspective of the outage performance of MUs served by a FAP, there is no difference between MUs and FUs covered by that FAP. Therefore, statistically,  \eqref{eq:SIR-m-fap} and \eqref{eq:I_macro} also describe the performance experienced by FUs of $a_f$.

\begin{remark}
While it might seem that we have assumed the same attenuation factor $\alpha$ for all different links in the network, in fact, the results do not change in general where  the path-loss exponent of outdoor and cross-wall (outdoor-indoor) transmissions  are assumed to be equal ($\alpha$) and larger than the path-loss exponent of indoor transmissions ($\beta$). To observe this, note that in \eqref{eq:SIR-m-fap}, exponent $\beta$ only affects users in $\Uc_f(a_f)$. However, due to our power control assumption, terms like $d^{\beta}(u_f,a_f)$ do not appear in the interference $I_{m,f}$. The same is true for our analysis presented in the next section corresponding to MUs served by the MBS.
\end{remark}

Consider  FAP $a_f\in\Ac_f$ positioned at distance $d_f=d(a_f,b_m)$ from MBS $b_m$. In the rest of this section,  we derive  upper  and lower bounds on the outage probability of MU  $u_m\in\Uc_m(a_f)$ as a function of $d_f$. As just mentioned, the same bounds hold for FUs covered by $a_f$, as well.

For MU $\uh_m\in\Uc_m^{(-f)}$,  since  $\uh_m$ is directly  serviced by MBS $b_m$, instead of one of the FAPs such as $a_f$, we must have $\kappa d(\uh_m,b_m)\leq d_{\uh_m}^{(f)} \leq d(\uh_m,a_f).$ Therefore, ${d(\uh_m,b_m) \over d(\uh_m,a_f)} \leq {1\over \kappa}.$

Let $\delta_{\uh_m}\triangleq {d(\uh_m,b_m) \over d(\uh_m,a_f)}$, where, as a reminder,  $b_m$ and $a_f$ denote the base station and the FAP at distance $d_f$ from $b_m$, respectively. As we just argued, $\delta_{\uh_m}\leq \kappa^{-1}$, for all $\uh_m\in\Uc_m^{(-f)}$. Given the complicated distribution of $\delta_{\uh_m}$, in order to characterize the outage probability, we quantize $\delta_{\uh_m}$.

\begin{figure}[t]
\begin{center}
\includegraphics[width=7cm]{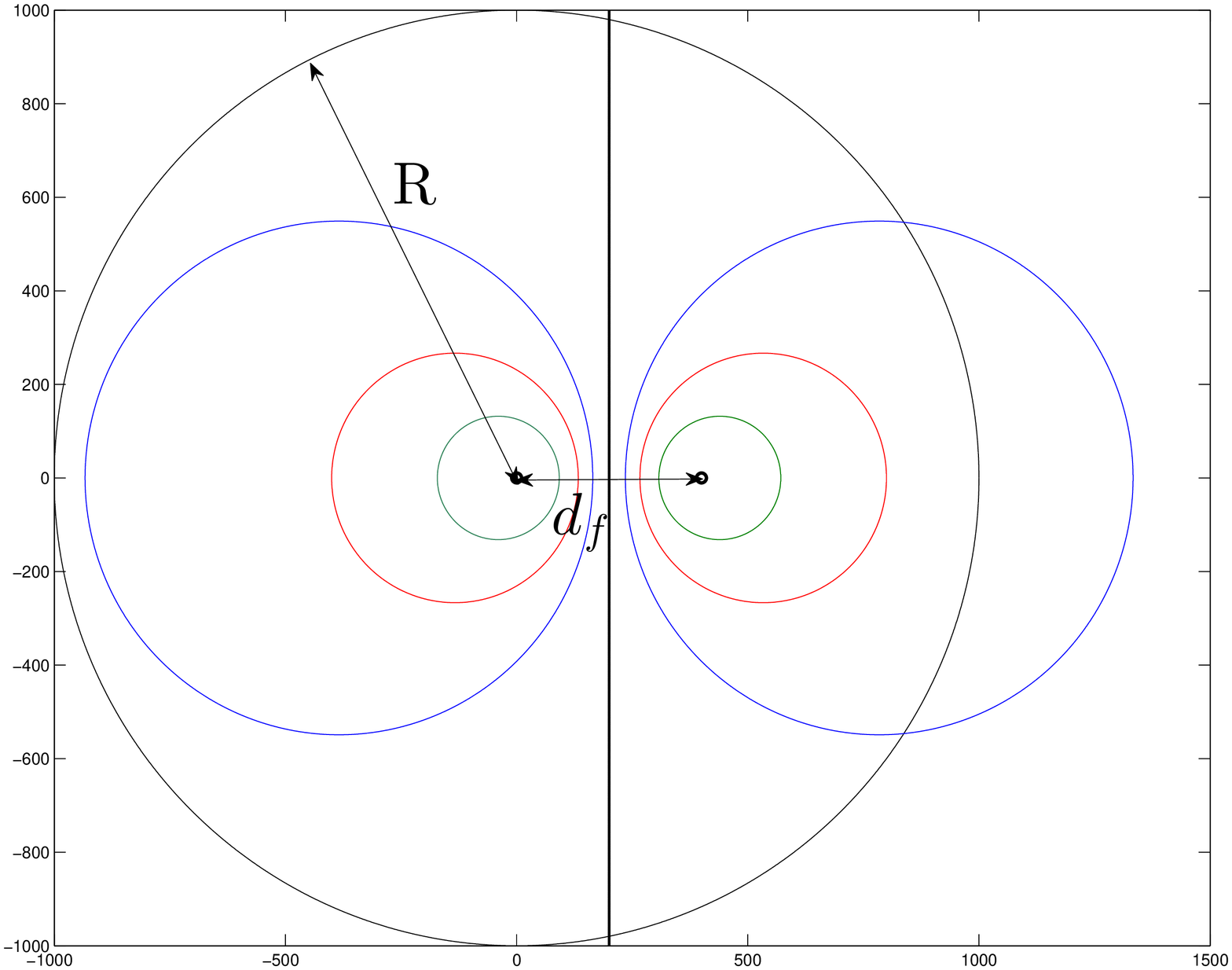}\caption{Partitioning the coverage area}\label{fig:kappa-quant}\end{center}
\end{figure}

Consider the setup shown in Fig.~\ref{fig:kappa-quant}, where FAP $a_f$ is located at distance $d_f$ from $b_m$. The coverage area of $b_m$ is shown by the black circle of radius $R$ centered at $b_m$. The green circle on the right represents  points with $\delta_{\uh_m}=\kappa^{-1}$. Similarly, the points on the green circle on the left have $\delta_{\uh_m}=\kappa$. The other pairs of circles correspond to some other values of $\kappa'>\kappa$. Points on the black line have $\delta_{\uh_m}=1$. Note that, by our assumption, all MUs are located inside the black circle. Therefore, the parts of colored circles that are outside of the black circle have zero probability. Consider $\kappa_0\triangleq\kappa<\kappa_1<\ldots<\kappa_{t-1}<\kappa_{t}\triangleq 1$. Let
\begin{align}
\hat{\delta}^{u}_{\uh_m}=
\left\{
\begin{array}{ccl}
\kappa_i^{-1},  &   &  {\rm if \;\;}\kappa_{i+1}^{-1}<{\delta}_{\uh_m}\leq \kappa_{i}^{-1}, \\
\kappa_{i+1}  &   &  {\rm if \;\;} \kappa_{i}<{\delta}_{\uh_m}\leq \kappa_{i+1},\\
\kappa  &   &  {\rm if \;\;} {\delta}_{\uh_m}\leq \kappa,
\end{array}
\right.\label{eq:delta-u}
\end{align}
and
\begin{align}
\hat{\delta}^{l}_{\uh_m}=
\left\{
\begin{array}{ccl}
\kappa_{i+1}^{-1},  &   &  {\rm if \;\;}\kappa_{i+1}^{-1}<{\delta}_{\uh_m}\leq \kappa_{i}^{-1}, \\
\kappa_{i}  &   &  {\rm if \;\;} \kappa_{i}<{\delta}_{\uh_m}\leq \kappa_{i+1},\\
0  &   &  {\rm if \;\;} {\delta}_{\uh_m}\leq \kappa.\label{eq:delta-l}
\end{array}
\right.
\end{align}
Then by construction,  $\hat{\delta}^{l}_{\uh_m}\leq {\delta}_{\uh_m}\leq \hat{\delta}^{u}_{\uh_m}$, for all $\uh_m\in\Uc_m^{(-f)}$, and,  unlike ${\delta}_{\uh_m}$, $\hat{\delta}^u_{\uh_m}$ and $\hat{\delta}^l_{\uh_m}$ are finite-alphabet random variables. Let $\Sc_i$ and $s_i$, $i=1,\ldots,t$, denote the region  corresponding to  $(\hat{\delta}^u_{\uh_m},\hat{\delta}^l_{\uh_m})=(\kappa_{i-1}^{-1},\kappa_{i}^{-1})$, and its area, respectively. Similarly, define $\Sc_{i}$  and $s_{i}$, $i=-t,\ldots,-1$ to correspond  to the region with $(\hat{\delta}_{\uh_m}^u,\hat{\delta}^l_{\uh_m})=(\kappa_{-i},\kappa_{-i-1})$. Finally, $\Sc_{0}$  and $s_{0}$ correspond to $(\hat{\delta}^u_{\uh_m},\hat{\delta}^l_{\uh_m})=(\kappa,0)$. In Appendix A, given $t\in\mathds{N}^+$,  we present analytic expressions for computing   $s_{-t},s_{-t+1},\ldots,s_{t-1},s_t$.

\begin{lemma}\label{lemma:delta-hat-dist}
For $i=1,\ldots,t$, $p_i\triangleq \P(\hat{\delta}^u_{\uh_m}=\kappa_{i-1}^{-1})=\P(\hat{\delta}^l_{\uh_m}=\kappa_{i}^{-1})={s_i \over s}$, for $i=-t,\ldots,-1$, $p_{i}\triangleq\P(\hat{\delta}^u_{\uh_m}=\kappa_{-i})=\P(\hat{\delta}^l_{\uh_m}=\kappa_{-i-1})={s_{i} \over s}$,
and $p_0\triangleq\P(\hat{\delta}^u_{\uh_m}=\kappa)=\P(\hat{\delta}^l_{\uh_m}=0)={s_0\over s}$,
where  $s\triangleq  s_0+\sum_{i=1}^t(s_i+s_{-i})$.
\end{lemma}
\vspace{0.2cm}
\begin{proof}
Let $\Sc$ denotes the whole circuit of radius $R$ minus the coverage area of $a_f$. Hence the area of $\Sc$  is equal to $s= s_0+\sum_{i=1}^t(s_i+s_{-i})$. For $i=-t,\ldots,-1$, we have
\begin{align*}
\P(\hat{\delta}^u_{\uh_m}\hspace{-.3em}=\hspace{-.2em}\kappa_{-i}\hspace{-.em})\hspace{-.1em}\hspace{-.2em}=\E[\hspace{-.1em}\ind_{\uh_m\in \Sc_{i}}\hspace{-.1em}]\hspace{-.2em}=\E\hspace{-.2em}\left[\E[\ind_{\uh_m\in \Sc_{i}}|\Ac_f]\right]\hspace{-.2em}=\hspace{-.2em}\E\hspace{-.2em}\left[\hspace{-.2em}
{ s_{i}-S_{i,m}\over s- S_{m}}\hspace{-.2em}\right]\hspace{-.2em},
\end{align*}
where $\Sc_{m}$ and $S_{m}$  denote the region in $\Sc$ that is covered by the MBS $b_m$,  and its area, respectively. (This is of course the area that is  not covered by FAPs.) Also, $S_{i,m}$ denotes the area of $\Sc_{i}\cap \Sc_{m}$. Note that  $S_{i,m}$ and  $S_{m}$ are both random variables that depend on the locations of the FAPs. To  derive the desired result, we employ the tower property one more time:
\begin{align*}
\P(\hat{\delta}_{\uh_m}=\kappa_i)&=\E\left[\E\Big[{ s_{i}-S_{i,m}\over s- S_{m}}|S_{m}\Big]\right]=\E\left[{ s_{i}-(s_{i}/s)S_{m}\over s- S_{m}}\right]\nonumber\\
&={s_{i}\over s}.
\end{align*}
The proof of the rest of the theorem follows from the same argument.
\end{proof}
\vspace{0.2cm}

 Let $P_{\rm out}^{m,f}(d_f)$ denote the outage probability experienced by a MU covered by a FAP located at distance $d_f$ of $b_m$. We employ Lemma \ref{lemma:delta-hat-dist} and our upper-bounding and lower-bounding quantizations  of $\delta_{\uh_m}$  to derive the following theorem that presents both an upper bound and a lower bound on  $P_{\rm out}^{m,f}(d_f)$.

\begin{theorem}\label{thm:3}
 Let $T_h\triangleq {T\over n_h}$. For $t\in\mathds{N}^+$, define
\begin{align}
q_l(s)\triangleq {p_0\over 1+s\kappa^{\alpha}   /\eta}+\sum_{i=1}^{t}\Big({p_i \over 1+s /(\eta\kappa_{i-1}^{\alpha})} +{p_{-i} \over 1+s\kappa_{i}^{\alpha}   /\eta}\Big),\label{eq:q-u}
\end{align}
\begin{align}
q_u(s)\triangleq p_0+\sum_{i=1}^{t}\Big({p_i \over 1+s  /(\eta\kappa_{i}^{\alpha})} +{p_{-i} \over 1+s\kappa_{i-1}^{\alpha}  /\eta}\Big),\label{eq:q-l}
\end{align}
where  $(p_{-t},\ldots,p_t)$ are defined and characterized   in Lemma \ref{lemma:delta-hat-dist}, $\eta\triangleq{P_f/ P_m}$, and  $\tau_o=\tau(-\log q_l({T_h\over \sigma^2}))$, with $\tau(\cdot)$ defined in \eqref{eq:tau-s}. Then,
\begin{align*}
P_{\rm out}^{m,f}(d_f) \leq 1&-{(1+T_h)(\ex^{\bar{n}_{\rm mu}^f/(1+T_h)}-1)\over \ex^{\bar{n}_{\rm mu}^f}-1}\nonumber\\
&\;\;\;\;\cdot \ex^{-\bar{n}_{\rm fu}T_h/(1+T_h)-\bar{n}_{\rm mu} (1-q_u(T_h/\sigma^2))}
\end{align*}
and
\begin{align*}
&P_{\rm out}^{m,f}(d_f)  \geq  1-{(1+T_h)(\ex^{\bar{n}_{\rm mu}^f/(1+T_h)}-1)\over \ex^{\bar{n}_{\rm mu}^f}-1}\ex^{-\bar{n}_{\rm fu}T_h/(1+T_h)} \nonumber\\
&\cdot \hspace{-.2em}\Big(\hspace{-.1em}{\ex^{(\hspace{-.1em}\bar{n}_{\rm mu}-\bar{n}_{\rm mu}^f)(q_l\hspace{-.1em}({T_h\over \sigma^2})-1)+\bar{n}_{\rm fap}\hspace{-.1em}
(\tau_o\hspace{-.1em}-1)} \over (1-\ex^{-\bar{n}_{\rm fap}} )\tau_o} \hspace{-.2em} +\hspace{-.2em}{\ex^{\g^{-1}\hspace{-.1em}-\bar{n}_{\rm fap} \hspace{-.1em}+ \g^{-1}\hspace{-.1em}\log (\hspace{-.1em}\g \bar{n}_{\rm fap}\hspace{-.1em})} \over 1-\ex^{-\bar{n}_{\rm fap}}}\hspace{-.2em}\Big).
\end{align*}
\end{theorem}
\begin{proof}
The details of the proof is presented  in Appendix~\ref{proof-th3}, but the outline of the proof is as follows. First,  we derive upper and lower bounds on the upload SIR experienced  by macro user $u_m\in\Uc_m(a_f)$, namely $\SIR_{m,f}$. To achieve this goal, we employ  the quantizations of $\delta_{\uh_m}$ defined in \eqref{eq:delta-u} and \eqref{eq:delta-l}. Then, we connect the outage probability with the Laplace transform of the bounds on SIR.  Finally, we use the fact that the support sets of the locations of MUs served by MBS and FAPs do not overlap to prove that  independence of the number of interfering users in different groups.
\end{proof}

\begin{remark}
In Theorem \ref{thm:3}, $t$ corresponds to the number of  partition levels of the MBS's coverage area, and is a parameter that can be selected arbitrarily. In other words, the bounds hold for any $t\in\mathds{N}^+$, but choosing higher value of $t$ leads to tighter upper and lower bounds.
\end{remark}
%


\begin{remark}\label{remark:3}
In Theorem \ref{thm:3}, the terms in the bounds that depend on  distance $d_f$ are $\bar{n}_{\rm mu}^f$, $q_u(T_h/\sigma^2)$ and $q_l(T_h/\sigma^2)$. As we will see in the numerical results presented in Section \ref{sec:numerical}, both the upper and lower bounds are not monotonic in $d_f$. This follows from the non-monotonic behavior of $q_u(T_h/\sigma^2)$ and $q_l(T_h/\sigma^2)$. The term $(1+T_h)(\ex^{\bar{n}_{\rm mu}^f/(1+T_h)}-1)/( \ex^{\bar{n}_{\rm mu}^f}-1)$, which appears in both upper and lower bound, is usually very close to one, and in monotonically decreasing in $d_f$. Therefore, the other terms in each bound are the dominant terms.
\end{remark}

\begin{remark}
To gain more insight on the effect of different parameters on the bounds in Theorem \ref{thm:3}, we can consider their approximate values for typical set of parameters, when the number of carriers is large. As mentioned in Remark \ref{remark:3},  $(1+T_h)(\ex^{\bar{n}_{\rm mu}^f/(1+T_h)}-1)/( \ex^{\bar{n}_{\rm mu}^f}-1)$ is  close to one, especially if the number of carriers is large. We can also approximate  $\tau_o$ as $\frac{1-\ex^{-\log( q_l({T_h\over \sigma^2}))\gamma \bar{n}_{\rm mu}}}{\log (q_l({T_h\over \sigma^2}))\gamma \bar{n}_{\rm mu}}$ by approximating  of $\ex^{-s}$  in $\tau(s)$ as $1-s$. Employing these approximations, and ignoring the other non-dominant terms, the upper and lower bound can be simplified as
$1-\ex^{-\bar{n}_{\rm fu}T_h-\bar{n}_{\rm mu} (1-q_u(T_h/\sigma^2))}$
and
$ 1-\ex^{-\bar{n}_{\rm fu}T_h-(1-q_l(T_h/\sigma^2))(\bar{n}_{\rm mu}- 0.5\gamma\bar{n}_{\rm fap} \bar{n}_{\rm mu} -\bar{n}_{\rm mu}^f )},$
respectively.
When the FAP gets close to the MBS, \ie $d_f\ll  R$, $\delta_{\uh_m}\approx 1$, and hence, $q_u(T_h/\sigma^2), q_l(T_h/\sigma^2) \approx \frac{1}{1+T_h/(\sigma^2 \eta)}$, or  $1-q_u(T_h/\sigma^2),1-q_l(T_h/\sigma^2) \approx T_h/(\sigma^2\eta)$, which further simplifies the upper and lower bounds to
$1-\ex^{-T_h(\bar{n}_{\rm fu}+\frac{\bar{n}_{\rm mu} }{\eta \sigma^2})}$,
and
$1-\ex^{-T_h(\bar{n}_{\rm fu}+\frac{\bar{n}_{\rm mu} }{\eta \sigma^2}-\frac{\bar{n}_{\rm mu}^f }{\eta \sigma^2}- \frac{\gamma\bar{n}_{\rm fap} \bar{n}_{\rm mu} }{2\eta \sigma^2})}
$, respectively.
\end{remark}

\subsection{MU served by the MBS}
The upload SIR experienced  by user $u_m\in\Uc_m^{(-f)}$ at the MBS $b_m$ in   subband $i$ is equal to
\begin{align}
\SIR_{m,m}&= {\frac{P_m|h^i_{u_m,b_m}|^2}{n_s} \over I_{m,m}},\label{eq:SIR-m-m}
\end{align}
where, for $i\in\{1,\ldots, n_s\}$,
\begin{align}
 I_{m,m}\hspace{-.2em}&=\hspace{-.4em}\sum_{a_f\in\Ac_f}\hspace{-.2em}\sum\limits_{\uh_m\in \Uc_m(a_f)}\Big({d(\uh_m,a_f)\over d(\uh_m,b_m)}\Big)^{\a}  |h^i_{\uh_m,b_m}|^2\E[\ind_{c_{\uh_m}[i]=c_{u_m}[i]}]  \frac{P_f}{n_s}\nonumber\\
&\hspace{3em}+\sum\limits_{\uh_m\in \Uc_m^{(-f)}\backslash u_m}\frac{P_m}{n_s}|h^i_{\uh_m,b_m}|^2 \E[\ind_{c_{\uh_m}[i]=c_{u_m}[i]}]\nonumber\\
 &=\sum_{a_f\in\Ac_f}\sum\limits_{u_m\in \Uc_m(a_f)}\Big({d(u_m,a_f)\over d(u_m,b_m)}\Big)^{\a} |h^i_{u_m,b_m}|^2 \frac{P_f}{G}\nonumber\\
&\hspace{3em}+  \sum\limits_{\uh_m\in \Uc_m^{(-f)}\backslash u_m}\frac{P_m}{G}|h^i_{\uh_m,b_m}|^2.\label{eq:I-m-m}
\end{align}
In deriving \eqref{eq:I-m-m}, we have ignored the interference caused by the FUs. The reason is that in
 most cases the distance between  FU $u_f$ and its  FAP  is  much  smaller than the distance between $u_f$ and $b_m$.
\begin{theorem}\label{thm:4}
Let $P_{\rm out}^{m,m}\triangleq \P(\SIR_{m,m}<T)$. Then, $P_{\rm out}^{m,m}\geq  1-(1+T_h) (\ex^{-\bar{n}_{\rm mu}T_h/(1+T_h)+\bar{n}_{\rm fap}(\tau'_o-1)}+\ex^{\g^{-1}-\bar{n}_{\rm fap} + \g^{-1}\log (\g \bar{n}_{\rm fap})} )$, and
\[
P_{\rm out}^{m,m}\; \leq\; 1-{(1+T_h)(\ex^{\bar{n}_{\rm mu}/(1+T_h)}-1) \over \ex^{\bar{n}_{\rm mu}}-1}
\]
where $\tau'_o={\ex^{\gamma \bar{n}_{\rm mu}T_h/(1+T_h)}-1 \over  \gamma \bar{n}_{\rm mu}T_h/(1+T_h)}$.

\end{theorem}
\begin{proof}
The proof is relegated to Appendix~\ref{proof-th4}. Similar to the proof of Theorem 3, here too, we derive upper and lower bound on the experienced SIR, $\SIR_{m,m}$.  In this case, when a macro user is served by the MBS,  $d(u_m,a_f)\leq \kappa d(u_m,b_m)$. Employing this bound, yields a lower bound on $\SIR_{m,m}$.  To derive the upper bound, we only consider the interference caused  by the other MUs served by the MBS.
\end{proof}

\begin{remark}
Typically, the  second term in the upper bound  is negligible  compared to the first term and can be ignored. Approximating by the first-order Taylor expansion, $\tau'_o$ can be approximated as $\tau'_o\approx 1+0.5\gamma \bar{n}_{\rm mu}T_h$, which holds for  large number of carriers per subband. Using these  approximations, the lower bound can be simplified to $1-\ex^{-\bar{n}_{\rm mu}T_h(1-\frac{1}{2}\gamma\bar{n}_{\rm fap})}$.
\end{remark}

\section{Numerical results}\label{sec:numerical}
In this section,  to investigate the uplink network performance and to verify our upper and lower bounds, we present some simulation results. Monte-Carlo computer simulations with $10^5$ realizations are carried out to validate our analytical bounds and illustrate the accuracy of our approximations. The considered scenario  is a two-tier network in a circle of radius $R=1\, {\rm Km}$ with the MBS located at the center. In the ensuing plots, we use the default values in Table \ref{Tab.1}, unless otherwise stated. Fig. \ref{fig:model} shows a realization of the network with the specified  parameters. Note that the size of the coverage area of a femtocell depends on its distance from the MBS. In our model, if a MU falls in the coverage areas of more than one user, it is serviced by the closest one.

\begin{figure}[h]
\begin{center}
\includegraphics[height=7cm]{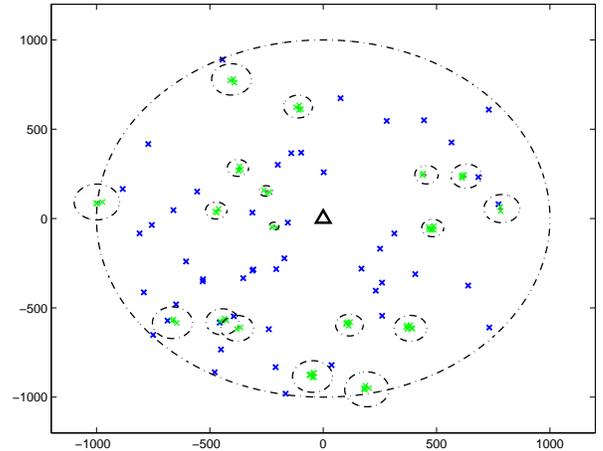}\caption{Sample realization of the network with the parameters specified in Table I. (Blue x: MU, green x: FU, circles (except for the largest one): coverage area of a FAP.)}\label{fig:model}\end{center}
\end{figure}

\begin{table}[b]
\footnotesize
\begin{center}
\caption{Simulations parameters}\label{Tab.1}
\begin{tabular}{|c|c|c|}
  \hline
  Sym.         & Description                              & Default Values \\
\hline \hline
 $\lambda_f$        &   density of FAPs  &  $5\times 10^{-6} \, {\rm m^{-2}}$   \\
\hline
 $\mu_m$            &   density of macrocell users     &    $15\times 10^{-6}\, {\rm m^{-2}}$      \\
\hline
 $\mu_f$            &     density of FUs   &    $0.01\, {\rm m^{-2}}$     \\
 \hline
 $\Delta$             &    ring width of FUs placement   &    $5\, {\rm m}$     \\
 \hline
 $R_f$             &ring  internal radius of  FUs placement&    $10\, {\rm m}$     \\
\hline
$\alpha$             &     path loss exponent  &   4     \\
\hline
$d_f$             &    distance between considered FAP and MBS   &    $700\, {\rm m\;}$     \\
\hline
$T$             &     SIR threshold level  &    $2$     \\
\hline
$n_s$             &   number of subbands    &    $32$     \\
\hline
$n_h$             &   number of subchannels in each subbands    &    $256$     \\
\hline
$\eta$             &    power ratio between FAPs and MBS   &    $25$     \\
\hline
$\kappa$             &    handover parameter   &    $0.1$     \\
\hline
\hline
    \end{tabular}
\end{center}
\end{table}

 Figs.~\ref{versus_T_FAP} and \ref{versus_T_MBS} show the conditional outage probabilities of   MUs serviced by a FAP located at $d_f=700\, {\rm m\;}$ from $b_m$, and the average outage probabilities  of MUs served by the MBS, respectively. Different curves in these figures  correspond to different MUs densities.  Figs.~\ref{versus_T_FAP} and \ref{versus_T_MBS} reflect that our analytical upper and lower bounds (solid and dot curves) are reasonable approximations for all considered SIR thresholds.

As expected, increasing the threshold level increases the probability of outage. Clearly this does not imply that  the performance can be improved by lowering  $T$, as  its reduction decreases the achieved  rate as well. In general, there is a trade-off between \emph{expected capacity} \cite{EffrosG:98 ,EffrosG:10} and threshold $T$. The problem of maximizing the expected rate by optimizing  $T$ is  studied  in \cite{ZeinalpourJ:13} for  a single-tier MCFH system . We leave extending those results to multi-tier networks for  future research.

\begin{figure}[t]
\begin{center}
\includegraphics[height=6cm]{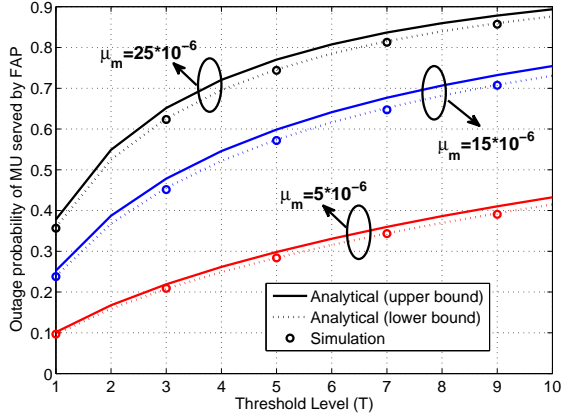}\caption{Conditional outage probability of  MUs served by FAPs located at distance $d_f=700 {\rm m}$ from the MBS versus threshold level $(T)$ for different MU densities.}\label{versus_T_FAP}\end{center}
\end{figure}

\begin{figure}[t]
\begin{center}
\includegraphics[height=6cm]{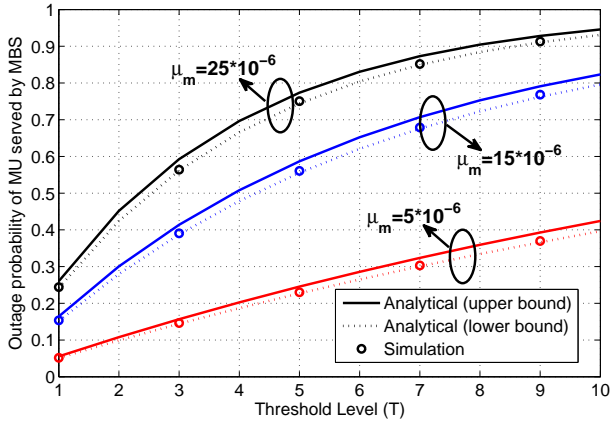}\caption{Average outage probability of  MUs served by MBS versus threshold level $(T)$ for different MUs densities.}\label{versus_T_MBS}\end{center}
\end{figure}

\begin{figure}[t]
\begin{center}
\includegraphics[height=6cm]{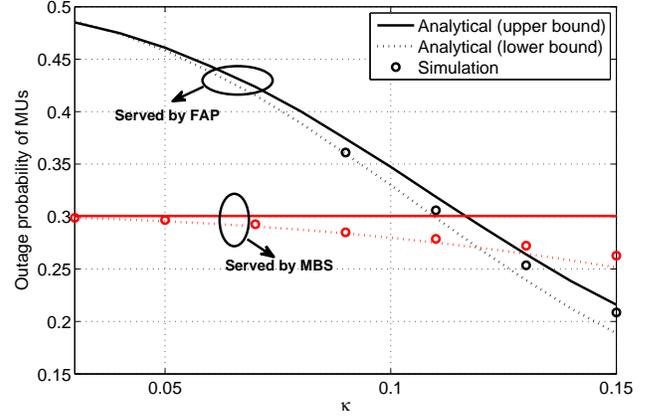}\caption{Average outage probability of MUs versus $\kappa$.}\label{versus_kappa}\end{center}
\end{figure}

\begin{figure}[h]
\begin{center}
\includegraphics[height=6cm]{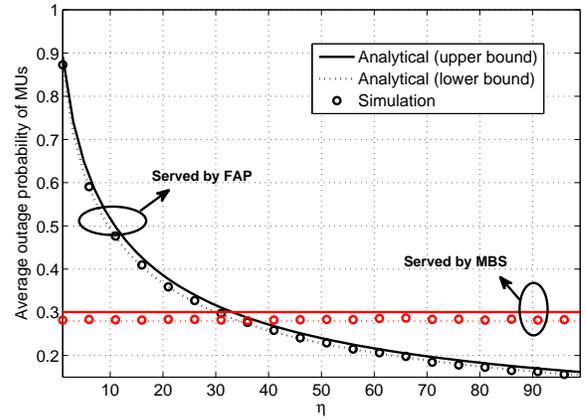}\caption{Average outage performance of MUs versus the power ratio between FAPs and MBS ($\eta={P_f/ P_m}$).}\label{versus_eta}\end{center}
\end{figure}

Fig.~\ref{versus_kappa} demonstrates how  the average outage probabilities of MUs vary with handover parameter $\kappa$. Here too the bounds are consistent with  the simulation results, but the gap increases slightly as $\kappa$ increases. In contrast to the downlink scenario \cite{CheungQ:12}, where  the outage probability is not monotonic in $\kappa$, here, increasing $\kappa$ improves the performance for both MUs and FUs. This result is consistent with \cite{XiaC:10}, where the authors argue  that  in non-orthogonal setups, open access is strictly better than closed access policy.  The  difference between uplink and downlink arises from the fact that in the downlink scenario as the MUs get farther away from the MBS, their received powers and hence SIRs decrease. On the other hand,  in the uplink scenario, \ as they become farther away from the MBS,  due to power control, their transmit powers  increase as well to compensate for the path loss. Naturally,  increasing the handover parameter leads to  more MUs being covered  by FAPs and hence  to lower  co-tier interference.

Note that for plotting the average probability experienced by  MUs served by  FAPs, we have taken the expected value of the upper and lower bounds mentioned in Theorem \ref{thm:3} by considering the randomness in $d_f$.

Fig.~\ref{versus_eta} shows the average outage performance of MUs as a function of  $\eta={P_f/ P_m}$, the power ratio between FAPs and MBS. In these plots we have fixed the transmit power of MBS and because of this, the outage curves of the MUs served by the MBS are almost constant. Obviously the outage of MUs served by FAPs improves by increasing FAPs transmit powers. Note that although increasing  FAPs powers  increases the interferences level, but its effect is not significant for MUs, whose performance is mainly limited by other MUs and not FUs.

Fig.~\ref{versus_df} illustrates the conditional outage probability of MUs served by a FAP,  as a function of the FAP's normalized distance from the MBS. As it can be observed from the figure, at first, the outage probability  increases as the MU gets farther  form the MBS. In fact because of the assumption of constant received power by the MBS in the uplink scenario, as the MU gets farther from the MBS, it will transmit
at a higher power, which leads to the  degradation in the performance of FUs and also MUs served by FAPs.  However, as the femtocells get close to the fringes of the cell, their users outage probabilities start to  improve as well. The reason is that  femtocells that are far away from the MBS have larger coverage areas and therefore, in those regions most  MUs are serviced by nearby FAPs, which results in lower interference caused by them.

\begin{figure}[t]
\begin{center}
\includegraphics[height=6cm]{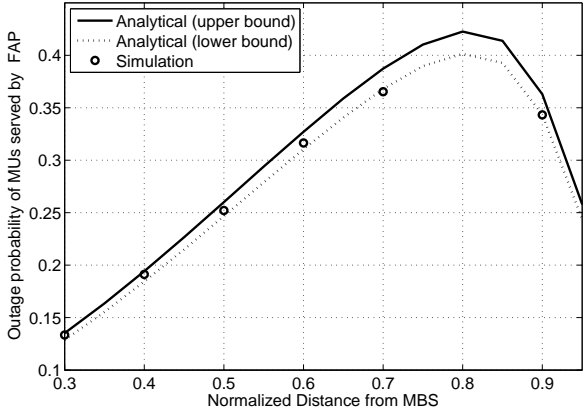}\caption{Conditional  outage probability of a MU served by a FAPs as
a function of the normalized distance of FAP from the MBS.}\label{versus_df}\end{center}
\end{figure}

\begin{figure}[t]
\begin{center}
\includegraphics[height=6cm]{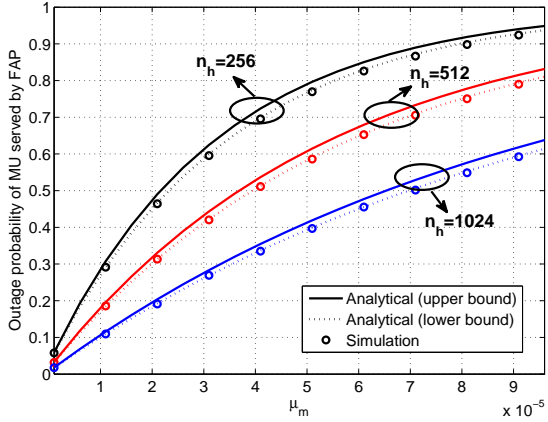}\caption{Conditional outage probability of MUs served by a FAP located at distance $d_f=700 {\rm m}$ from the MBS versus their density.}\label{versus_mu_m1}\end{center}
\end{figure}

\begin{figure}[t]
\begin{center}
\includegraphics[height=6cm]{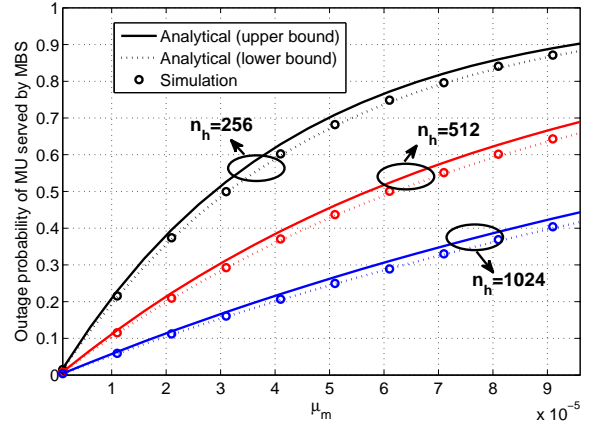}\caption{Average outage probability of MUs served by MBS versus their density.}\label{versus_mu_m2}\end{center}
\end{figure}

Figs.~\ref{versus_mu_m1} and \ref{versus_mu_m2} show the conditional outage probability of MUs served by a FAP located at  $d_f=700\, {\rm m\;}$,  and the average outage probability of MUs served by the MBS, respectively, as a function of MUs density $\mu_m$. Obviously,  increasing the macrocell users density will increase their  outage probabilities as well, because of more co-tier interferences. However their performance can be greatly improved by boosting the number of available sub-channels as can be seen in the figure, too.

\section{Conclusions}\label{sec:conclusion}
In this paper we investigated the uplink performance of  two-tier networks consisting of a macrocell overlaid by femtocells. We considered a stochastic spatial distribution for MUs, FUs and FAPs, and assumed that they are generated by independent PPPs.  For cell association, we considered an open access policy, where each MU is assigned to its nearest FAP if their distance is less than its distance from the MBS times some factor $\kappa<1$. Under this model we studied the outage performance of the system and derived analytical upper and lower bounds on the outage probabilities of both FUs and MUs. The bounds were shown to be tight by our simulations.

Throughout the paper we considered a fixed threshold $\kappa$ for all MUs. A more general model is when $\kappa$ is not fixed and depends on the FAP. In other words, since $\kappa$ determines the coverage area of FAPs, it is conceivable to consider a scenario where  FAPs are heterogeneous  and can choose their  coverage areas.  For instance, a  FAP can move toward a closed access policy by lowering its corresponding $\kappa$, \ie by only  accepting MUs that are very close.

FAPs are  connected to a central gateway via wired connections. The capacity constraints imposed by this backhaul wired network can potentially affect the cell selection procedure and may impede some femtocells to service all MUs that fell in their coverage area. Characterizing  the effect of this constraint on the system's performance is another  interesting question that is left for future research.

\section*{Acknowledgments}
The authors would like to thank the anonymous  reviewers for their helpful comments and suggestions,   especially  one of the reviewers who pointed us to reference \cite{FossZ:96}.

\setcounter{equation}{0}
\renewcommand{\theequation}{\thesection.\arabic{equation}}

\appendices

\section{Derivation of $s_i$}
As defined in Section \ref{sec:outage}, $s_i$, $i=1,\ldots,t$, denotes the area of the region  corresponding to  $\hat{\delta}^u_{\uh_m}=\kappa_{i-1}^{-1}$ and $\hat{\delta}^l_{\uh_m}=\kappa_{i}^{-1}$ or equivalently  $\kappa_i^{-1}\leq\delta_{\uh_m}\leq \kappa_{i-1}^{-1}$. We showed earlier that the points that satisfy $\delta_{\uh_m}=\kappa^{-1}$, are located on a circle of radius $\frac{\kappa d_f}{1-\kappa^2}$ centered at $(\frac{d_f}{1-\kappa^2},0)$. Hence, $s_i$ is the macrocell coverage area surrounded by two such  circles with $\kappa=\kappa_i$ and $\kappa=\kappa_{i-1}$. Therefore, for $i=1,2,...,t-1$,
\begin{align*}
s_i = {\pi d_f^2 \kappa_i^2 \over (1-\kappa_i^2)^2}-\ind_{\kappa_{i}>1-{d_f \over R}}f({\kappa_i d_f \over {1-\kappa_i^2}},R, {d_f \over {1-\kappa_i^2}})-\sum_{j=0}^{i-1}s_j,
\end{align*}
 and $s_t=R^2(\theta -{1 \over 2}\sin(2\theta))-\sum_{j=0}^{t-1}s_j$, where $\theta=\arccos({d_f \over 2R})$ and

\begin{align*}
f (a,b,c)&\triangleq a^2 \sec^{-1}({2ac\over b^2-a^2-c^2})-b^2 \sec^{-1}({2bc\over b^2+c^2-a^2})\nonumber\\
&+{1\over 2}\sqrt{(a+b+c)(b+c-a)(c+a-b)(a+b-c)}
\end{align*}
\cite{Weisstein}. Similarly,  $s_{i}$, for $i=-t,\ldots,-1$ denotes the area  of the region where $\kappa_{-i-1}\leq\delta_{\uh_m}\leq\kappa_{-i}$, and $\hat{\delta}^l_{\uh_m}=\kappa_{-i-1}$. Here, the geometrical location of the points satisfying  $\delta_{\uh_m}=\kappa$ is a circle with the origin at  $(\frac{-\kappa^2d_f}{1-\kappa^2},0)$ and  radius $\frac{\kappa d_f}{1-\kappa^2}$. Therefore the area of the macrocell zone enclosed by two such the circles with $\kappa=\kappa_{-i}$ and $\kappa=\kappa_{-i-1}$ is given by $s_i \hspace{-.2em} =\hspace{-.2em} {\pi d_f^2 \kappa_{-i}^2 \over (1-\kappa_{-i}^2)^2}-\ind_{\kappa_{-i}>{R \over R+d_f}}f({\kappa_{-i} d_f \over {1-\kappa_{-i}^2}},R, {d_f \kappa_{-i}^2\over {1-\kappa_{-i}^2}}) -\sum_{j=i+1}^{0} s_j,$
for $i=-t+1,...,-1$, and $s_{-t}= R^2(\pi-\theta +{1 \over 2}\sin(2\theta))-\sum_{j=-t+1}^{0}s_j$, where again $\theta=\arccos({d_f \over 2R})$.

\section{Proof of Theorem \ref{thm:3}}\label{proof-th3}
Define event $\Ec=\{d(a_f,b_m)=d_f, N_m^{a_f}\geq 1\}$. Then,  by definition, $P_{\rm out}^{m,f}(d_f)= \P(\SIR_{m,f}<T|  \Ec)$
Combining  the quantizations defined in \eqref{eq:delta-u} and \eqref{eq:delta-l} with  \eqref{eq:SIR-m-fap} and \eqref{eq:I_macro}, we  have
\begin{align}
{n_h|h^i_{u_m,a_f}|^2 \over  K_u } \leq \SIR_{m,f}& \leq  {n_h|h^i_{u_m,a_f}|^2 \over  K_l }\label{eq:lb-SIR-mu-to-fap}
\end{align}
where
\begin{align}
K_u  \triangleq&  \sum\limits_{u_f\in \Uc_f(a_f )} |h^i_{u_f,a_f}|^2  + \sum\limits_{ \uh_m\in \Uc_m(a_f)\backslash u_m} |h^i_{\uh_m,a_f}|^2 \nonumber\\
&+ {1 \over\eta } \sum\limits_{\uh_m\in \Uc_m^{(-f)}} \hspace{-1em}|h^i_{\uh_m,a_f}|^2( \hat{\delta}_{\uh_m}^u )^{\alpha}, \label{eq:K-u}
\end{align}
and
\begin{align}
K_l \triangleq & \sum\limits_{u_f\in \Uc_f(a_f )} |h^i_{u_f,a_f}|^2  + \sum\limits_{ \uh_m\in \Uc_m(a_f)\backslash u_m} |h^i_{\uh_m,a_f}|^2\nonumber\\
& + {1 \over\eta } \sum\limits_{\uh_m\in \Uc_m^{(-f)}} \hspace{-1em}|h^i_{\uh_m,a_f}|^2( \hat{\delta}_{\uh_m}^l )^{\alpha}, \label{eq:K-l}
\end{align}

From \eqref{eq:lb-SIR-mu-to-fap}, we have
\begin{align}
P_{\rm out}^{m,f}(d_f) &\leq \P \Big({n_h|h^i_{u_m,a_f}|^2 \over  K_u }<T| \Ec \Big)\nonumber\\
&= \E \left[\E\Big[ \ind_{|h^i_{u_m,a_f}|^2 <   {K_u T\over n_h}}|  \Ec, K_u \Big]\right]\nonumber\\
&\overset{{\mathrm{(a)}}}{=}1-\E\Big[  {\rm e}^{-{T K_u  \over n_h{\sigma}^2}}\Big|\Ec \Big]\nonumber\\
&=1-\Phi_{K_u}({T  \over n_h{\sigma}^2}|d_f),\label{P-out-upper-steps}
\end{align}
where $(a)$ follows from our assumption that  channel coefficient $|h^i_{u_m,a_f}|$ has a Rayleigh distribution. Here $\Phi_{K_u}(s|d_f)\triangleq \E[\ex^{-s K_u}|\Ec]$ denotes the conditional Laplace transform of  $K_u$ defined in \eqref{eq:K-u}.
Since conditioned on the number of  users in each category, the channel coefficients and $\{\hat{\delta}_{\uh_m}\}_{\uh_m\in\Uc_m^{(-f)}}$ are all independent of each other,  it follows that
\begin{align}
\Phi_{K_u}\hspace{-.2em}(\hspace{-.1em}s|d_f\hspace{-.1em})&\hspace{-.2em}=\hspace{-.2em}\E \hspace{-.2em}
\Big[\hspace{-.2em} \left(\hspace{-.2em}{1 \over 1+s\sigma^2}\hspace{-.2em}\right) ^ {\hspace{-.2em}N_f^{a_f}+N_m^{a_f}-1}\hspace{-1em}\cdot q_u(s)^{N_m^{b_m}}\Big|\Ec\ \hspace{-.2em}\Big].\label{phi-K-u}
\end{align}
where $q_u(s)$ is defined in \eqref{eq:q-u}.
To derive the lower bound, again from  \eqref{eq:lb-SIR-mu-to-fap}, and following a similar steps as in \eqref{P-out-upper-steps}, we derive
\begin{align}
P_{\rm out}^{m,f}(d_f)&\geq  \P \Big({n_h|h^i_{u_m,a_f}|^2 \over  K_l }<T|  \Ec\Big)=1-\E\Big[  {\rm e}^{-{T K  \over n_h{\sigma}^2}}\Big|\Ec\ \Big]\nonumber\\
&=1-\Phi_{K_l}({T  \over n_h{\sigma}^2}|d_f),\label{P-out-lower-steps}
\end{align}
where $\Phi_{K_l}(s|d_f)\triangleq \E[\ex^{-s K_l}|\Ec]$. Also, as argued before in deriving \eqref{phi-K-u}, we have
\begin{align}
\Phi_{K_l}\hspace{-.2em}(\hspace{-.1em}s|d_f\hspace{-.1em})&\hspace{-.2em}=\hspace{-.2em}\E \hspace{-.2em}
\Big[\hspace{-.2em} \left(\hspace{-.2em}{1 \over 1+s\sigma^2}\hspace{-.2em}\right) ^ {\hspace{-.2em}N_f^{a_f}+N_m^{a_f}-1}\hspace{-1em}\cdot q_l(s)^{N_m^{b_m}}\Big|\Ec\ \hspace{-.2em}\Big],
\end{align}
where  $q_l(s)$ is defined in \eqref{eq:q-l}.

Conditioned on the location of $a_f$, $N_f^{a_f}$, $N_m^{a_f}$ and $N_m^{b_m}$ are independent random variables. The independence of $N_f^{a_f}$ and  $(N_m^{a_f},N_m^{b_m})$ follows from our initial assumption that the process of drawing  MUs and FUs are independent. To see the independence of $N_m^{a_f}$ and $N_m^{b_m}$, note that $\Uc_m(a_f)$ denotes the users that are located in a circle of radius $d_f\sqrt{(1-\kappa^2)^{-2}-1}$. (Refer to Fig.~\ref{fig:MU-covered-FAP}.) On the other hand,  conditioned on the location of $a_f$, $\Uc_m^{(-f)}$ denotes users that are not located in any of the circles corresponding to different  FAPs, one of which is the mentioned  circle corresponding to $a_f$. Therefore, the support sets of the  locations of MUs in  $\Uc_m^{(-f)}$ and the MUs in $\Uc_m(a_f)$ do not have any overlap. Therefore, since the macro  users are generated by a PPP process, $N_m^{a_f}$ and $N_m^{b_m}$ are independent random variables. As a result,
\begin{align}
\Phi_{K_u}(s|d_f)=& \Phi_{N_f^{a_f}}(\log(1+s\sigma^2))
\Phi^+_{N_m^{a_f}}(\log (1+s\sigma^2)|d_f)\nonumber\\
&\cdot \Phi_{N_m^{b_m}}(-\log q_u(s)|d_f).\label{eq:Ku-der}
\end{align}
Similarly,
\begin{align}
\Phi_{K_l}(s|d_f)=& \Phi_{N_f^{a_f}}(\log(1+s\sigma^2))) \Phi^+_{N_m^{a_f}}(\log (1+s\sigma^2)|d_f)\nonumber\\
&\cdot \Phi_{N_m^{b_m}}(-\log q_l(s)|d_f).\label{eq:Kl-der}
\end{align}
Combining the bounds derived for $\Phi_{N_m^{b_m}}(s|d_f)$ in Theorem \ref{thm:2} with  \eqref{eq:Phi-Uf-af}, \eqref{eq:Phi-Um-af}, \eqref{P-out-upper-steps}, \eqref{P-out-lower-steps}, \eqref{eq:Ku-der} and \eqref{eq:Kl-der} completes the proof of Theorem \ref{thm:3}.

\section{Proof of Theorem \ref{thm:4}}\label{proof-th4}
For user $u_m\in\Uc_m(a_f)$, by our assignment policy, we should have $d(u_m,a_f)\leq \kappa d(u_m,b_m)$. Therefore, $I_{m,m}$ can be upper-bounded as
\begin{align}
I_{m,m}& \hspace{-.3em}\leq \hspace{-.5em}\sum_{a_f\in\Ac_f}\hspace{-.1em}\sum\limits_{\uh_m\in \Uc_m(a_f)}\hspace{-1.2em}\kappa^{\a} |h^i_{\uh_m,b_m}|^2\frac{P_f}{G}+ \hspace{-1em} \sum\limits_{\uh_m\in \Uc_m^{(-f)}\backslash u_m}\hspace{-1.3em}\frac{P_m}{G}|h^i_{\uh_m,b_m}|^2\nonumber\\
& \leq  \sum\limits_{\uh_m\in \Uc_m \backslash u_m}\frac{P_m}{G}|h^i_{\uh_m,b_m}|^2,
\end{align}
where the last line follows since $\kappa<1$. Also, clearly,  $I_{m,m} \geq   \sum_{\uh_m\in \Uc_m^{(-f)}\backslash u_m}\frac{P_m}{G}|h^i_{\uh_m,b_m}|^2.$
Hence,
\[
{n_h |h^i_{u_m,b_m}|^2\over \bar{K}_u} \leq \SIR_{m,m} < {n_h |h^i_{u_m,b_m}|^2\over \bar{K}_l},
\]
where
$\bar{K}_l=\sum_{\uh_m\in \Uc_m^{(-f)}\backslash u_m}|h^i_{\uh_m,b_m}|^2 $ and $\bar{K}_u=\sum_{\uh_m\in \Uc_m \backslash u_m}|h^i_{\uh_m,b_m}|^2$, and
\[
1-\Phi_{\bar{K}_l}({T  \over n_h\sigma^2}) \leq P_{\rm out}^{m,m} \leq 1-\Phi_{\bar{K}_u}({T  \over n_h\sigma^2}).
\]
To derive the upper bound,  note that $\Phi_{\bar{K}_u}(s)=\E[\ex^{-s\bar{K}_u}]=\E[(\E[\ex^{-s|h^i|^2}])^{|\Uc_m|-1}||\Uc_m|\geq 1]
={\ex^{a\bar{n}_{\rm mu}}-1\over a(\ex^{\bar{n}_{\rm mu}}-1)},$
 where $a=1/(1+s\sigma^2)$. For the lower bound, we employ the upper bound on $\Phi_{\Uc_m^{(-f)}}$ presented in Theorem \ref{thm:1}. The only difference here compared to Theorem \ref{thm:1} is that here we need to condition on $N_m^{b_m}\geq 1$. However, $\E\left[\left.\ex^{-sN_m^{b_m}} \right|N_m^{b_m}\geq 1\right]\leq \E[\ex^{-sN_m^{b_m}}].$
The reason is that for any positive integer-valued random variable $X$ with $p_i=P(X=i)$, $i=0,1,\ldots$, we have  $\E[\ex^{-X}]=p_0+\sum_{i=1}^{\infty} p_i\ex^{-i}$ and $\E[\ex^{-X}|X\geq 1]=(1-p_0)^{-1}\sum_{i=1}^{\infty} p_i\ex^{-i}$. Therefore, $\E[\ex^{-X}]-E[\ex^{-X}|X\geq 1]=p_0(\sum_{i=1}^{\infty}p_i(1-\ex^{-i}))/(1-p_0)\geq 0$. Combining this with Theorem \ref{thm:1} yields the lower bound.
\bibliographystyle{unsrt}
\bibliography{myrefs}

\begin{thebibliography}{10}

\bibitem{CheungQ:12}
W.C. Cheung, T.Q.S. Quek, and M.~Kountouris.
\newblock Throughput optimization, spectrum allocation, and access control in
  two-tier femtocell networks.
\newblock {\em Sel. Areas in Comm., IEEE Journal on}, 30(3):561--574, 2012.

\bibitem{JoS:12}
H.S. Jo, Y.J. Sang, P.~Xia, and J.G. Andrews.
\newblock Heterogeneous cellular networks with flexible cell association: A
  comprehensive downlink {SINR} analysis.
\newblock {\em Wireless Communications, IEEE Trans. on}, 11(10):3484--3495,
  2012.

\bibitem{BaccelliZ:97}
F.~Baccelli and S.~Zuyev.
\newblock Stochastic geometry models of mobile communication networks.
\newblock {\em Frontiers in queueing}, pages 227--243, 1997.

\bibitem{BaccelliK:97}
F.~Baccelli, M.~Klein, M.~Lebourges, and S.~Zuyev.
\newblock Stochastic geometry and architecture of communication networks.
\newblock {\em Telecommunication Systems}, 7(1-3):209--227, 1997.

\bibitem{Brown:00}
T.X. Brown.
\newblock Cellular performance bounds via shotgun cellular systems.
\newblock {\em Sel. Areas in Comm., IEEE Journal on}, 18(11):2443--2455, 2000.

\bibitem{HaenggiA:09}
M.~Haenggi, J.~G. Andrews, F.~Baccelli, O.~Dousse, and M.~Franceschetti.
\newblock Stochastic geometry and random graphs for the analysis and design of
  wireless networks.
\newblock {\em Sel. Areas in Comm., IEEE Journal on}, 27(7):1029--1046, 2009.

\bibitem{BaccelliM:09}
F.~Baccelli, P.~Miihlethaler, and B.~Blaszczyszyn.
\newblock Stochastic analysis of spatial and opportunistic {Aloha}.
\newblock {\em Sel. Areas in Comm., IEEE Journal on}, 27(7):1105--1119, 2009.

\bibitem{AndrewsB:09}
J.G. Andrews, F.~Baccelli, and R.K. Ganti.
\newblock A tractable approach to coverage and rate in cellular networks.
\newblock {\em Comm., IEEE Trans. on}, 59(11):3122--3134, Nov. 2011.

\bibitem{DhillonG:12}
H.S. Dhillon, R.K. Ganti, F.~Baccelli, and J.G. Andrews.
\newblock Modeling and analysis of k-tier downlink heterogeneous cellular
  networks.
\newblock {\em Sel. Areas in Comm., IEEE Journal on}, 30(3):550--560, 2012.

\bibitem{Mukherjee:12}
S.~Mukherjee.
\newblock Distribution of downlink sinr in heterogeneous cellular networks.
\newblock {\em Sel. Areas in Comm., IEEE Journal on}, 30(3):575--585, 2012.

\bibitem{WangQ:12}
Wang~Chi Cheung, T.Q.S. Quek, and M.~Kountouris.
\newblock Throughput optimization, spectrum allocation, and access control in
  two-tier femtocell networks.
\newblock {\em Sel. Areas in Comm., IEEE Journal on}, 30(3):561--574, 2012.

\bibitem{AndrewsC:12}
J.G. Andrews, H.~Claussen, M.~Dohler, S.~Rangan, and M.C. Reed.
\newblock Femtocells: Past, present, and future.
\newblock {\em Sel. Areas in Comm., IEEE Journal on}, 30(3):497--508, 2012.

\bibitem{ChandrasekharA:09-power}
V.~Chandrasekhar, J.G. Andrews, Tarik Muharemovict, Zukang Shen, and Alan
  Gatherer.
\newblock Power control in two-tier femtocell networks.
\newblock {\em Wireless Comm., IEEE Trans. on}, 8(8):4316--4328, 2009.

\bibitem{JoM:09}
H.S. Jo, C.~Mun, J.~Moon, and J.G. Yook.
\newblock Interference mitigation using uplink power control for two-tier
  femtocell networks.
\newblock {\em Wire. Comm., IEEE Trans. on}, 8(10):4906--4910, 2009.

\bibitem{ChandrasekharA:09-tcom}
V.~Chandrasekhar and J.G. Andrews.
\newblock Spectrum allocation in tiered cellular networks.
\newblock {\em Comm., IEEE Trans. on}, 57(10):3059--3068, 2009.

\bibitem{ChaZ:13}
D.~Cao, S.~Zhou, and Z.~Niu.
\newblock Improving the energy efficiency of two-tier heterogeneous cellular
  networks through partial spectrum reuse.
\newblock {\em Wire. Comm., IEEE Trans. on}, 12(8):4129--4141, 2013.

\bibitem{XiangZ:10}
J.~Xiang, Y.~Zhang, T.~Skeie, and L.~Xie.
\newblock Downlink spectrum sharing for cognitive radio femtocell networks.
\newblock {\em Sys. Journal, IEEE}, 4(4):524--534, 2010.

\bibitem{ShiH:10}
Y.~Shi, Y.~T. Hou, H.~Zhou, and S.~F. Midkiff.
\newblock Distributed cross-layer optimization for cognitive radio networks.
\newblock {\em Veh. Tech., IEEE Trans. on}, 59(8):4058--4069, 2010.

\bibitem{ChandrasekharA:08}
V.~Chandrasekhar, J.~Andrews, and A.~Gatherer.
\newblock Femtocell networks: a survey.
\newblock {\em Comm. Mag., IEEE}, 46(9):59--67, 2008.

\bibitem{ClaussenH:08}
H.~Claussen, Lester~TW Ho, and L.~G. Samuel.
\newblock An overview of the femtocell concept.
\newblock {\em Bell Labs Tech. Journal}, 13(1):221--245, 2008.

\bibitem{NovlanD:13}
T.~D. Novlan, H.~S. Dhillon, and J.~G. Andrews.
\newblock Analytical modeling of uplink cellular networks.
\newblock {\em Wireless Comm., IEEE Trans. on}, 12(6):2669--2679, 2013.

\bibitem{ChandrasekharA:09}
V.~Chandrasekhar and J.~G. Andrews.
\newblock Uplink capacity and interference avoidance for two-tier femtocell
  networks.
\newblock {\em Wireless Comm., IEEE Trans. on}, 8(7):3498--3509, 2009.

\bibitem{XiaC:10}
P.~Xia, V.~Chandrasekhar, and J.~G. Andrews.
\newblock Open vs. closed access femtocells in the uplink.
\newblock {\em Wireless Comm., IEEE Trans. on}, 9(12):3798--3809, 2010.

\bibitem{ChakchoukH:12}
N.~Chakchouk and B.~Hamdaoui.
\newblock Uplink performance characterization and analysis of two-tier
  femtocell networks.
\newblock {\em Veh. Tech., IEEE Trans. on}, 61(9):4057--4068, 2012.

\bibitem{YuM:12}
B.~Yu, S.~Mukherjee, H.~Ishii, and L.~Yang.
\newblock Dynamic tdd support in the lte-b enhanced local area architecture.
\newblock In {\em Globecom Workshops (GC Wkshps), IEEE}, pages 585--591. IEEE,
  2012.

\bibitem{BaoL:13}
W.~Bao and B.~Liang.
\newblock Uplink interference analysis for two-tier cellular networks with
  diverse users under random spatial patterns.
\newblock In {\em Proc. of IEEE/CIC Int. Conf. on Comm. in China (ICCC), XiÕan,
  China}, 2013.

\bibitem{BaoL:13-acm}
Wei Bao and Ben Liang.
\newblock Understanding the benefits of open access in femtocell networks:
  Stochastic geometric analysis in the uplink.
\newblock In {\em Proc. of the 16th ACM Int. Conf. on Mod., Ana. \& Sim. of
  Wire. and Mob. Sys.}, pages 237--246. ACM, 2013.

\bibitem{ElSawyH:14}
H.~ElSawy and E.~Hossain.
\newblock On stochastic geometry modeling of cellular uplink transmission with
  truncated channel inversion power control.
\newblock {\em arXiv preprint arXiv:1401.6145}, 2014.

\bibitem{Lawrey99}
E.~Lawrey.
\newblock Multiuser {OFDM}.
\newblock In {\em Signal Processing and Its Applications, Int. Symp. on},
  volume~2, pages 761--764, 1999.

\bibitem{LanceK:97}
E.~Lance and G.K. Kaleh.
\newblock A diversity scheme for a phase-coherent frequency-hopping
  spread-spectrum system.
\newblock {\em Communications, IEEE Transactions on}, 45(9):1123--1129, 1997.

\bibitem{EbrahimiN:04}
M.~Ebrahimi and M.~Nasiri-Kenari.
\newblock Performance analysis of multicarrier frequency-hopping ({MC-FH})
  code-division multiple-access systems: Uncoded and coded schemes.
\newblock {\em Vehicular Technology, IEEE Transactions on}, 53(4):968--981,
  2004.

\bibitem{YazdiN:04}
Z.~Zeinalpour-Yazdi and M.~Nasiri-Kenari.
\newblock Performance comparison of coherent and non-coherent multicarrier
  frequency-hopping code division multiple-access systems.
\newblock In {\em IEEE Int. Symp. on Per., Ind. and Mob. Rad. Comm.}, volume~1,
  pages 165--169. IEEE, 2004.

\bibitem{TaghaviN:03}
Z.~Taghavi and M.~Nasiri-Kenari.
\newblock Multiuser performance analysis of {MC-FH} and {FFH} systems in the
  presence of partial-band interference.
\newblock In {\em IEEE Proc. on Personal, Indoor and Mobile Radio
  Communications}, volume~1, pages 578--582. IEEE, 2003.

\bibitem{NikjahB:08}
R.~Nikjah and N.C. Beaulieu.
\newblock On antijamming in general cdma systems-part ii: Antijamming
  performance of coded multicarrier frequency-hopping spread spectrum systems.
\newblock {\em Wireless Communications, IEEE Trans. on}, 7(3):888--897, 2008.

\bibitem{EffrosG:98}
M.~Effros and A.~Goldsmith.
\newblock Capacity of general channels with receiver side information.
\newblock In {\em Proc. IEEE Int. Symp. Inform. Theory}, page~39, 1998.

\bibitem{EffrosG:10}
M.~Effros, A.~Goldsmith, and Y.~Liang.
\newblock Generalizing capacity: New definitions and capacity theorems for
  composite channels.
\newblock {\em IEEE Trans. Inform. Theory}, 56(7):3069--3087, 2010.

\bibitem{ZeinalpourJ:13}
Z.~Zeinalpour-Yazdi and S.~Jalali.
\newblock On expected capacity of multicarrier frequency hopping systems.
\newblock In {\em Comm. and Info. Theory (IWCIT), 2013 Iran Workshop on}, pages
  1--6. IEEE, 2013.

\bibitem{FossZ:96}
S.G. Foss and S.A. Zuyev.
\newblock On a voronoi aggregative process related to a bivariate poisson
  process.
\newblock {\em Advances in Applied Probability}, pages 965--981, 1996.

\bibitem{Weisstein}
Eric~W. Weisstein.
\newblock Lune.
\newblock {\em From MathWorld--A Wolfram Web Resource.
  http://mathworld.wolfram.com/Lune.html}.

\end{thebibliography}

\newpage

\end{document}